\documentclass[11pt]{article}

\newif\ifcomments
\commentstrue

\usepackage{geometry}
\usepackage[ruled,vlined,linesnumbered]{algorithm2e}
\usepackage{amsmath}
\usepackage{bm}
\usepackage{enumerate}
\usepackage{float}
\usepackage{graphicx}
\usepackage[utf8]{inputenc}
\usepackage[export]{adjustbox}
\usepackage{mathtools}
\usepackage[english]{babel}
\usepackage{amsthm}
\usepackage{amssymb}
\usepackage{enumitem}
\usepackage{fontawesome5}

\usepackage{thm-restate,color,xspace}
\usepackage{comment}
\usepackage{thmtools}
\usepackage{xcolor}
\usepackage{nameref}
\usepackage[T1]{fontenc}

\definecolor{ForestGreen}{rgb}{0.1333,0.5451,0.1333}
\definecolor{DarkRed}{rgb}{0.65,0,0}
\definecolor{Red}{rgb}{1,0,0}

\usepackage[linktocpage=true,
pagebackref=true,colorlinks,
linkcolor=DarkRed,citecolor=ForestGreen,
bookmarks,bookmarksopen,bookmarksnumbered]
{hyperref}
\usepackage[nameinlink]{cleveref}

\declaretheorem[numberwithin=section]{theorem}
\declaretheorem[numberlike=theorem]{lemma}

\declaretheorem[numberlike=theorem]{corollary}

\declaretheorem[numberlike=theorem]{claim}

\declaretheorem[numberlike=theorem,style=definition]{definition}

\DeclareMathAlphabet{\mathdutchcal}{U}{dutchcal}{m}{n}
\newcommand{\Aextend}{\ensuremath{A^{\mathsf{ext}}}}
\newcommand{\Bextend}{\ensuremath{B^{\mathsf{ext}}}}

\newcommand{\mincut}{\mathrm{mincut}}
\newcommand{\enumcuts}{\textsc{EnumerateCuts}\xspace}
\newcommand{\enumcutshelp}{\textsc{EnumerateCutsHelp}\xspace}

\newcommand{\locmincut}{\textsc{LocalMinCut}\xspace}
\newcommand{\phisparsify}{\textsc{$\phi$-Sparsify}\xspace}
\newcommand{\algslow}{\textsc{SparsifySlow}\xspace}
\newcommand{\algfast}{\textsc{SparsifyFast}\xspace}
\newcommand{\expdec}{\textsc{ExpanderDecompose}\xspace}
\newcommand{\polylog}{\mathrm{polylog}}

\let\oldnl\nl%
\newcommand{\nonl}{\renewcommand{\nl}{\let\nl\oldnl}}%

\SetKwRepeat{Do}{do}{while}%
\SetKwInOut{effect}{Effect}
\SetKw{Continue}{continue}

\graphicspath{ {./images/} }
\newcommand{\T}{\mathcal{T}}
\renewcommand{\S}{\mathcal{S}}
\newcommand{\C}{\mathcal{C}}

\newcommand{\vol}{\mathrm{vol}}
\global\long\def\Otil{\tilde{O}}
\global\long\def\poly{\mathrm{poly}}
\newcommand{\J}{\mathcal{J}}
\newcommand{\Ginc}{G^{\textsf{inc}}}
\newcommand{\Vinc}{V^{\textsf{inc}}}
\newcommand{\Einc}{E^{\textsf{inc}}}
\newcommand{\Vsplit}{V_{\textsf{split}}}
\newcommand{\Asplit}{A_{\textsf{split}}}
\newcommand{\Bsplit}{B_{\textsf{split}}}
\newcommand{\Csplit}{C_{\textsf{split}}}
\newcommand{\Tsplit}{\T_{\textsf{split}}}

\newcommand{\Vaux}{V^{\mathsf{aux}}}
\newcommand{\Eaux}{E^{\mathsf{aux}}}
\newcommand{\Gaux}{G^{\mathsf{aux}}}
\newcommand{\aux}{\mathsf{aux}}

\newcommand{\amcut}{$A$-minimal $(A,B)$-mincut\xspace}
\newcommand{\tmcut}{$A$-minimal $(A,\T\setminus A)$-mincut\xspace}
\newcommand{\tpar}{$(A,\T\setminus A)$}

\SetCommentSty{mycommfont}

\renewcommand{\subparagraph}{\paragraph}

\title{Vertex Sparsifiers for Hyperedge Connectivity}
\author{Han Jiang \and 
Shang-En Haung\and 
Thatchaphol Saranurak\thanks{University of Michigan} \and 
Tian Zhang}
\date{}

\begin{document}

\maketitle

\begin{abstract}
   
  Recently, Chalermsook et al.~{[}SODA'21{]} introduces a notion of \emph{vertex
  sparsifiers for $c$-edge connectivity, }which has found applications
  in parameterized algorithms for network design and also led to exciting
  dynamic algorithms for $c$-edge st-connectivity {[}Jin and Sun FOCS'21{]}. 
  
  We study a natural extension called \emph{vertex sparsifiers for $c$-hyperedge
  connectivity} and construct a sparsifier whose size matches the state-of-the-art
  for normal graphs. More specifically, we show that, given a hypergraph
  $G=(V,E)$ with $n$ vertices and $m$ hyperedges with $k$ terminal
  vertices and a parameter $c$, there exists a hypergraph $H$ containing
  only $O(kc^{3})$ hyperedges that preserves all minimum cuts (up to
  value $c$) between all subset of terminals. This matches the best
  bound of $O(kc^{3})$ edges for normal graphs by [Liu'20]. Moreover,
  $H$ can be constructed in almost-linear $O(p^{1+o(1)} + n(rc\log n)^{O(rc)}\log m)$ time where $r=\max_{e\in E}|e|$ is
  the rank of $G$ and $p=\sum_{e\in E}|e|$ is the total size of $G$, or in $\poly(m, n)$ time if we slightly relax the
  size to $O(kc^{3}\log^{1.5}(kc))$ hyperedges. 
  \end{abstract} 

\section{Introduction}
Graph sparsification has played a central role in graph algorithm research in the last two decades. 
Prominent examples include
spanners \cite{althofer1993sparse}, cut sparsifiers \cite{benczur2015randomized}, and spectral
sparsifiers \cite{spielman2011spectral}.
Recently, there has been significant effort in generalizing the
graph sparsification results to hypergraphs. For cut sparsifiers,
Kogan and Krauthgamer \cite{kogan2015sketching} generalized the Bencz{\'{u}}r
and Karger's cut sparsifiers~\cite{benczur2015randomized} by showing that, given any hypergraph
$G=(V,E)$ with $n$ vertices, there is a $(1+\epsilon)$-approximate
cut sparsifier $H$ containing $\Otil(nr/\epsilon^{2})$ hyperedges
where $r=\max_{e\in E}|e|$ denotes the \emph{rank} of the hypergraph.
After some follow-up work \cite{chekuri2018minimum,bansal2019new},
Chen, Khanna, and Nagda \cite{chen2020near} finally improved the
sparsifier size to $\Otil(n/\epsilon^{2})$ hyperedges, matching the
optimal bound for normal graphs. Another beautiful line of work generalizes
Speilman and Teng's spectral sparsifiers \cite{spielman2011spectral} to hypergraphs \cite{bansal2019new,soma2019spectral,kapralov2021towards}
and very recently results in spectral sparsifiers with $\Otil(n/\poly(\epsilon))$
hyperedges \cite{kapralov2022spectral}. We also mention that the
classical sparse connectivity certificates by Nagamochi and Ibaraki
\cite{nagamochi1992linear} were also generalized to hypergraphs by
Chekuri and Xu \cite{chekuri2018minimum}.

This paper studies a graph sparsification problem called \emph{vertex
sparsifiers for $c$-edge connectivity} recently introduced
by Chalermsook et al.~\cite{chalermsook2021vertex}. It is closely related to the vertex sparsifiers for edge cuts \cite{krauthgamer2013mimicking,khan2014mimicking} and vertex cuts \cite{KratschW20}. 
In this problem, we are given
an unweighted undirected graph $G=(V,E)$ and a set of terminals $\T\subseteq V$.
For any disjoint subsets $A,B\subseteq\T$, let $\mincut_{G}(A,B)$
denote the size of a minimum (edge-)cut that disconnects $A$ and
$B$. Now, a graph $H=(V_{H},E_{H})$ with $\T\subseteq V_{H}$ is
a \emph{$(\T,c)$-sparsifier} %
of $G$ if for any disjoint subsets $A,B\subseteq\T$, $\min\{c,\mincut_{G}(A,B)\}=\min\{c,\mincut_{H}(A,B)\}$.
Basically, $H$ preserves all minimum cut structures between the terminals
$\T$ up to the value $c$. This notion of graph sparsifiers has found
interesting applications in offline dynamic algorithms and network
design problems \cite{chalermsook2021vertex}. Moreover, the very
recent breakthrough on dynamic $c$-edge st-connectivity by Jin and
Sun \cite{jin2022fully} is also crucially based on dynamic algorithms
for maintaining $(\T,c)$-sparsifiers. 

In the original paper by \cite{chalermsook2021vertex}, they showed
that, for any graph $G=(V,E)$ and terminal set $\T$ of size $k$,
there exists a $(\T,c)$-sparsifier containing $O(kc^{4})$ edges (which can be constructed in $O(m(c\log n)^{O(c)})$ time) and also showed fast algorithms for constructing $(\T,c)$-sparsifiers
of size $k\cdot O(c)^{2c}$ in $mc^{O(c)}\log^{O(1)}{n}$ time. Then, Liu \cite{Liu20} improved
the size bound to $O(kc^{3})$ together with polynomial-time algorithms
(no exponential dependency on $c$) for constructing $(\T,c)$-sparsifiers
with $O(kc^{3}\log^{1.5}n)$ edges. 

A natural question is then whether these results can be extended to hypergraphs. The notion of $(\T,c)$-sparsifiers
itself can be naturally extended to hypergraphs by allowing $G$
and $H$ to be hypergraphs and letting $\mincut_{G}(A,B)$ denote
the value of the minimum hyperedge-cut instead. However, it is conceivable
that there might not exist a $(\T,c)$-sparsifier with $\poly(k, c)$.
This bound might require bad
dependency on the rank $r$, for example. 

In this paper, we show that the state-of-the-art for normal graphs indeed extend to hypergraphs and we can even slightly improve the bounds:
\begin{theorem}\label{thm:main}
Let $G=(V,E)$ be a hypergraph with $n$ vertices, $m$ hyperedges, rank $r$ and total size $p$.
Let $\T\subseteq V$ be the set of $k$ terminals. There are algorithms for computing the following:
\begin{enumerate}
\item a $(\T,c)$-sparsifier $H$ of $G$ with $O(kc^{3})$ hyperedges in $O(p^{1+o(1)} + n(rc\log n)^{O(rc)}\log m)$ time, which is almost-linear in the input size when both $r$ and $c=O(1)$, and
\item a $(\T,c)$-sparsifier $H$ of $G$ with $O(kc^{3}\log^{1.5}(kc))$ hyperedges
in $\poly(m, n)$ time.
\end{enumerate}
\end{theorem}

The first result matches the best known
bound of $O(kc^{3})$ edges for normal graphs \cite{Liu20}. When $r=O(1)$, the first time bound slightly improves the $O(m(c\log n)^{O(c)})$ bound of \cite{chalermsook2021vertex} for normal graphs. 
The second result removes the exponential dependency on $r$ and $c$ after relaxing the size by a $\log^{1.5}(kc)$ factor. The number of hyperedges in our sparsifier is completely independent from $n$, while the polynomial time algorithm by Liu \cite{Liu20} gives the size of $O(kc^{3}\log^{1.5}n)$. 
So this implies the first polynomial time construction of sparsifiers of size near-linear in $k$ and independent of $n$, even for normal graphs.

\subparagraph{Open Problems.}
Can we construct vertex sparsifiers for $c$-hyperedge connectivity
of $k\cdot\poly(c)$ size in near-linear time even when the rank is
unbounded? This is a prerequisite to near-linear
time algorithms for computing vertex sparsifiers for $c$-\emph{vertex
connectivity} of $k\cdot\poly(c)$ size. Such a result might
lead to dynamic $c$-vertex st-connectivity algorithm similar to the
previous development where a near-linear time construction of vertex
sparsifiers for $c$-edge connectivity leads to a dynamic algorithm
for $c$-edge st-connectivity \cite{jin2022fully}. As dynamic $c$-vertex
st-connectivity is one of the major open problems in dynamic graph
algorithms (known solutions only works for very small $c\le3$ \cite{eppstein1997sparsification,holm2001poly,peng2019optimal}), we view
this work as a stepping stone towards this goal.

\subsection{Technical Challenges}

There are two main obstacles that prevent us extending the results of \cite{Liu20,chalermsook2021vertex} directly from normal graphs to hypergraphs.
First, if we follow the divide and conquer framework of Chalermsook et al.~\cite{chalermsook2021vertex} in a straightforward way,
then we would end up with a much larger $(\T, c)$-sparsifier with $O(|\T|(rc)^3)$ hyperedges. This is because, in \cite{chalermsook2021vertex}, all vertices incident to the boundary edges are declared as new terminal vertices in the recursion.
However in our case, each hyperedge may contain $r$ vertices and this yields the dependency of $r$. 
To handle this issue, we instead introduce only two \emph{anchor vertices} for each boundary hyperedge. Our divide and conquer framework requires slightly more careful analysis, but this naturally gives a $(\T, c)$-sparsifier with $O(|\T|c^3)$ hyperedges.

The second obstacle is the near-linear time algorithm, Part (1) of \Cref{thm:main}. 
Chalermsook et al.~\cite{chalermsook2021vertex}  introduced \emph{auxiliary graphs} and apply the \phisparsify procedure on it to identify all \emph{essential} hyperedges, which roughly are hyperedges that will be kept in the sparsifier. 
However, there is a subtle small gap in \cite{chalermsook2021vertex}: their \phisparsify procedure could erroneously identify non-essential hyperedges as essential hyperedges.
This is explained in more detail in \Cref{fig:aux-graph-counterexample}. %
This bug results in a much larger $(\T, c)$-sparsifier. %
In this paper we fix the bug by (1) introducing a notion of \emph{useful} partitions of the terminal set and (2) providing an efficient algorithm that discards all non-useful partitions from the auxiliary graph.
Then, we show that, after our modification, this approach indeed gives a small $(\T, c)$-sparsifier as desired.

\subsection{Organization}

In \Cref{sec:preliminary} we review some basic definitions of hypergraphs. In \Cref{sec:structural-properties} we define contraction based $(\T, c)$-sparsifiers and introduce the divide and conquer framework.
In \Cref{subsection:Existence} we show the existence of a $(\T, c)$-sparsifier with $O(|\T|c^3)$ hyperedges.
In \Cref{subsection:Algorithm} we give a near-linear-time algorithm that computes a $(\T, c)$-sparsifier with  $O(|\T|c^3)$ hyperedges, proving Part (1) of \Cref{thm:main}.
At the end we prove the Part (2) of \Cref{thm:main} in \Cref{sec:part-2-of-main}.

\section{Preliminary}\label{sec:preliminary}

Let $G=(V, E)$ be a hypergraph. $V$ is the set of vertices and $E$ is a multiset of hyperedges with each hyperedge $e$ being a subset of $V$.
The \emph{rank} $r \coloneqq \max_{e\in E}|e|$ of a hypergraph is the size of the largest hyperedge, and the \emph{total size} $p\coloneqq \sum_{e\in E}|e|$ is the sum of all edge sizes.

For any two disjoint sets of vertices $A, B\subseteq V$,
let $E_G(A, B)$ denote the set of hyperedges with at least one endpoint in $A$ and at least one endpoint in $B$.
For any set of vertices $X\subseteq V$, we denote the \emph{boundary} of $X$ of the graph $G$ by 
$\partial_G X \coloneqq E_G(X, V\setminus X)$.
If the context is clear then we will omit the graph $G$ and write $\partial X$ instead.

\subparagraph{Restrictions and Induced Sub-Hypergraphs.}
Let  $\T\subseteq V\cup E$ be a mixed multiset of vertices and hyperedges,
for any set of vertices $X\subseteq V$,
we define the \emph{restriction} of the multiset $\T$ on $X$ to be
$\T|_X = (\T\cap X) \cup \{ e\cap X\ |\ e\in (\T\cap E) \text{ and } e\cap X\neq\emptyset \}$.
The \emph{induced sub-hypergraph} $G[X]$ is then defined over the vertex set $X$ with the restriction of all hyperedges $E|_X$, that is, $G[X] \coloneqq (X, E|_X)$.

\subparagraph{Incident Edges and Vertices.}
For any set of vertices $X\subseteq V$, define $E(X)$ to be the set of all hyperedges that incident to at least one vertex in $X$.
For any set of hyperedges $Y\subseteq E$, define $V(Y) = \bigcup_{y\in Y} y$ to be the set of vertices incident to hyperedges in $Y$. Similarly, for any mixed set of vertices and hyperedges $\T\subseteq V\cup E$ we define $V(\T)=(\T\cap V)\cup V(\T\cap E)$ to be the set of all vertices that are in the set or incident to any hyperedge in the set.

\section{Structural Properties on Hypergraphs} \label{sec:structural-properties}

In this section, we explore more structural properties on hypergraphs.
In particular,
we introduce \emph{anchored separated hyperedges}, 
and describe  useful properties in a divide and conquer framework that leads to a construction of $(\T, c)$-sparsifiers.

\subsection{Cuts in Hypergraphs}

Let $u$ and $v$ be two elements in $V$.
We say that $u$ and $v$ are \emph{connected} in a hypergraph $G$, if there is a path connecting $u$ and $v$.
Let $A, B\subseteq V$ be two disjoint sets of vertices.
$A$ and $B$ are \emph{disconnected} if for any $a\in A$ and $b\in B$, $a$ and $b$ are not connected.

\begin{definition}[Cuts and Minimum Cuts]\label{def:cuts-and-minimum-cuts}
A \emph{cut} is a bipartition $(X, V\setminus X)$ of vertices. The \emph{value} of the cut is $|\partial X| = |E_G(X, V\setminus X)|$.
For any disjoint subsets $A, B \subseteq V$,
if $A\subseteq X$ and $B\subseteq (V\setminus X)$ then we say that $(X, V\setminus X)$ is an $(A, B)$-cut.
A \emph{minimum $(A,B)$-cut} or \emph{$(A, B)$-mincut} is any $(A,B)$-cut with minimum value. Its value is denoted as $\mincut_G(A,B)$. Given a parameter $c$, a \emph{$c$-thresholded $(A, B)$-mincut cut value}
is defined as 
$$\mincut^c_G(A,B) \coloneqq \min(\mincut_G(A,B), c).$$
We usually write a \emph{$c$-thresholded $(A,B)$-mincut} to emphasize that the $(A,B)$-mincut has value at most $c$.
We say that a hyperedge $e$ is \emph{involved} in a cut $(X, V\setminus X)$ if $e\in E(X, V\setminus X)$.
\end{definition}

\subsection{$(\T, c)$-Equivalency and $(\T, c)$-Sparsifiers}

Our vertex sparsifier algorithms are based on identifying a set of hyperedges and contract them.
Given a hypergraph $G=(V, E)$ and a hyperedge $e\in E$, the \emph{contracted hypergraph} $G/e$ is defined by identifying all incident vertices $V(e)$ as one vertex, and then remove $e$ itself from the graph.
For any set of terminals $T\subseteq V$, the effect of contracting an hyperedge $e$ is denoted as $T/e$.
Similarly, for any set $\hat{E}\subseteq E$, we denote $G/\hat{E}$ the hypergraph obtained from $G$ by contracting all hyperedges in $\hat{E}$ (notice that all hyperedges in $\hat{E}$ are removed after the contraction.)

\begin{definition}[$(\T, c)$-Sparsifiers]
Let $G=(V_G, E_G)$ and $H=(V_H, E_H)$ be two hypergraphs.
Let $\T\subseteq V_G$ be the set of terminals.
We say $H$ is a \emph{contraction based $(\T, c)$-sparsifier} of $G$, 
if there exists a surjective (onto) projection $\pi: V_G\to V_H$, such that
for any $e\in E_H$ there is an edge $f\in E_G$ such that $\pi(f)= \cup_{v\in f}\{\pi(v)\} =e$,
and 
for any two subsets $T_1, T_2\subseteq \T$,
\[
\mincut_G^c(T_1, T_2) = \mincut_H^c(\pi(T_1), \pi(T_2)).
\]
Furthermore, if the terminals are not affected by the projection, i.e., $\pi(\T)=\T$, then we say that $G$ and $H$ are \emph{$(\T, c)$-equivalent}.
\end{definition}

\subparagraph{Remark.}
A more general $(\T, c)$-sparsifier would allow an arbitrary mapping $\pi$ on both vertices and edges. However, we note that all $(\T, c)$-sparsifiers constructed in this paper are always contraction based.
Therefore, for the ease of the presentation
we will omit the term ``contraction based''
when we mention $(\T, c)$-sparsifiers.

For the ease of the reading, we define the following set operations that allow us to add/remove hyperedges of $G$ into a sparsifier $H$.

\begin{definition}
Let $G=(V, E)$ be a hypergraph.
For any multiset $X$ of hyperedges over the vertices $V$,
and any contraction based $(\T, c)$-sparsifier $H$ with the projection $\pi$, define
\begin{itemize}[nosep]
    \item \emph{(Adding contracted hyperedges)} $H\cup X \coloneqq H\cup \pi(X)$, and
    \item \emph{(Removing contracted hyperedges)} $H-X \coloneqq H - \pi(X)$.
\end{itemize}
\end{definition}

\subsection{$(\T, c)$-Sparsifiers from a Divide and Conquer Framework}

Another important concept to our contraction based $(\T, c)$-sparsifier construction is that we apply a divide and conquer framework to $G$.
Let $G=(V, E)$ be a hypergraph and let $(V_1, V_2)$ be a bipartition of vertices.
We note that our divide and conquer framework is slightly different than just recurse on the induced sub-hypergraphs $G[V_1]$ and $G[V_2]$. In particular, for each \emph{separated hyperedge} $e$ we add two new \emph{anchor vertices} to  $e$, ended up slightly increasing the size of the vertex set  in the next-level recursion.

\newcommand{\varv}{\ensuremath\mathdutchcal{v}}

\newcommand{\Sep}{\mathit{Sep}}
\newcommand{\anchor}{\ensuremath\mathdutchcal{v}}
\newcommand{\GSep}{G^\mathsf{sep}}
\newcommand{\CSep}{C^\mathsf{sep}}
\subparagraph{Separated Hyperedges and Anchor Vertices.} 
Let $e\in E(V_1, V_2)$ be a hyperedge across the bipartition.
The \emph{separated hyperedges} of $e$ with respect to this bipartition $(V_1, V_2)$ is the set composed of hyperedges of $e$ restricted on both $V_1$ and $V_2$.
The \emph{anchored separated hyperedges} are separated hyperedges with additional \emph{anchor vertices}:
let $e_1 = e|_{V_1}$ and $e_2=e|_{V_2}$ be the separated hyperedges of $e$, then we introduce four new anchored vertices $\anchor_{e, 1}, \anchor_{e, 2}, \anchor_{e, 3}$, and $\anchor_{e, 4}$ and define $\hat{e}_1 := e_1\cup \{\anchor_{e, 1}, \anchor_{e, 2}\}$ and $\hat{e}_2 := e_2\cup \{\anchor_{e, 3}, \anchor_{e, 4}\}$.
Let $S=E(V_1, V_2)$ be the set of crossing hyperedges, in this paper the set of anchored separated hyperedges respect to bipartition $(V_1, V_2)$ are denoted by $\Sep(S, V_1, V_2) := \{\hat{e}_1, \hat{e}_2 \ |\ e\in S\}$.

Let $A_1=\{\anchor_{e,1}, \anchor_{e,2}\ |\ e\in E(V_1, V_2)\}$ and let $A_2 = \{\anchor_{e,3}, \anchor_{e,4}\ |\ e\in E(V_1, V_2)\}$ be the set of newly introduced anchor vertices.
These anchor vertices 
will be added to the terminal set in order to correctly preserve the mincut values.
That is, the terminal sets defined for the subproblems are $\T_1:=\T|_{V_1}\cup A_1$ and $\T_2:=\T|_{V_2}\cup A_2$.
Now, we define the anchored induced sub-hypergraphs, which are useful when applying the divide and conquer framework.

\begin{definition}[Anchored Induced Sub-Hypergraphs]
Let $G=(V, E)$ be a hypergraph and $V_1\subseteq V$ be a subset of vertices.
Define $V_2=V\setminus V_1$, $\GSep = G \cup \Sep(E(V_1, V_2), V_1, V_2) - E(V_1, V_2)$, and the set of anchored vertices to be $A_1 \cup A_2$.
Then,
the \emph{anchored induced sub-hyperegraph} for $V_1$ is  defined as $\hat{G}[V_1] := \GSep|_{V_1\cup A_1}$.
\end{definition}

\paragraph{The Divide and Conquer Framework.}
The most generic divide and conquer method works as the follows.
First, a bipartition $(V_1, V_2)$ are determined.
Then, the algorithm performs recursion on the anchored induced sub-hypergraphs $\hat{G}[V_1]$ and $\hat{G}[V_2]$ with terminal sets $\T_1$ and $\T_2$ respectively.
After obtaining the $(\T_1, c)$-sparsifier and $(\T_2, c)$-sparsifier from the subproblems, the algorithm combines them by replacing the anchored separated hyperedges with the original hyperedges.

\begin{algorithm}[H]
\caption{A Divide and Conquer Framework}\label{alg:divide-and-conquer}
\DontPrintSemicolon
\SetAlgoLined
\KwIn{Hypergraph $G$, terminal set $\T$, bipartition $(V_1, V_2)$ of vertices, parameter $c$.}
\KwOut {A $(\T, c)$-sparsifier $H$ for $G$.}
(\textbf{Divide}) Construct subproblems $(\hat{G}[V_1], \T_1)$ and $(\hat{G}[V_2], \T_2)$.\\
(\textbf{Conquer}) For $i\in \{1, 2\}$, obtain $H_i$, a $(\T_i, c)$-sparsifier of $\hat{G}[V_i]$.\label{line:conquer-step}\\
(\textbf{Combine}) Return $H := H_1\cup H_2\cup E(V_1, V_2) - \Sep(E(V_1, V_2), V_1, V_2)$.\\
\end{algorithm}

We summarize the divide and conquer framework in \Cref{alg:divide-and-conquer}.
The following lemma states the correctness of the framework.

\begin{restatable}[]{lemma}{Lemmacombinesparsifiers}
\label{lem:combine-sparsifiers}
$H$ returned from \Cref{alg:divide-and-conquer} is a $(\T, c)$-sparsifier.
\end{restatable}

\begin{proof}
Let $\pi_1: V_1\cup A_1\to V_{H_1}$ and $\pi_2: V_2\cup A_2\to V_{H_2}$ be the projection maps on $H_1$ and $H_2$ respectively.
Since $V_1\cup A_1$ and $V_2\cup A_2$ are disjoint, it is natural to define $\pi: V\to V_H$
by simply combining both maps where $\pi(v)=\pi_1(v)$ if $v\in V_1$, and $\pi(v)=\pi_2(v)$ if $v\in V_2$.

Now, it suffices to show that for any two disjoint subsets $A, B\subseteq\T$, we have $\mincut_G^c(A, B) = \mincut_H^c(\pi(A), \pi(B))$.

\noindent\textbf{Part 1.} We first show that $\mincut_G^c(A, B)\ge \mincut_H^c(\pi(A), \pi(B))$.
Let $(X, V\setminus X)$ be a minimum $(A, B)$-cut on $G$ with size $|\partial X|\le c$.
Intuitively, we will construct the cuts in the subproblems $\hat{G}[V_1]$ and $\hat{G}[V_2]$ using $(X, V\setminus X)$. Then we will argue that the preserved mincuts in $\hat{G}[V_1]$ and $\hat{G}[V_2]$ can be merged back, proving that there is a $(\pi(A), \pi(B))$-mincut in $H$ with size no larger than $|\partial X|$.

Let $S=E_G(V_1, V_2)$ be the set of hyperedges across the bipartition in the divide and conquer framework, and let $\hat{V}=V\cup A_1\cup A_2$ be the vertex set in $\GSep$.
We define the set of vertices $X^\mathsf{sep}$ that contains $X$ and all newly created anchor vertices that belongs to the $X$ side: for any $e\in S$, we add $\{\anchor_{e,1},\anchor_{e,2},\anchor_{e,3}, \anchor_{e,4}\}$ to $X^\mathsf{sep}$ if $e\subseteq X$ (the hyperedge is fully in the $X$ side).
We add $\{\anchor_{e,1},\anchor_{e, 3}\}$ to $X^\mathsf{sep}$ if $e\in S$. We add nothing if $e\subseteq V\setminus X$.

Now, we have $\partial X^\mathsf{sep} = (\partial X) \cup \Sep((\partial X)\cap S, V_1, V_2) - (\partial X)\cap S$. Moreover, $(X^\mathsf{sep}, \hat{V}\setminus X^\mathsf{sep})$ is an $(\Aextend, \Bextend)$-cut in $G^\mathsf{sep}$ of size $|\partial X|+|(\partial X)\cap S|$, where
\[
\begin{cases}
\Aextend := A\cup \{\anchor_{e, 1}, \anchor_{e, 3}\ |\ e\in ((\partial X)\cap S)\}, \text{ and }\\
\Bextend := B\cup \{\anchor_{e, 2}, \anchor_{e, 4}\ |\ e\in ((\partial X)\cap S)\}.
\end{cases}
\]
Intuitively, by carefully extend the pair $(A, B)$ to a larger pair $(\Aextend, \Bextend)$
we ensure that \emph{all} separated hyperedges $\Sep((\partial X)\cap S, V_1, V_2)$ appear in every $(\Aextend, \Bextend)$-mincut on $\GSep$. See \Cref{fig:divide-and-conquer-illustration}.

\begin{figure}[htbp]
\centering
\includegraphics[width=0.48\linewidth]{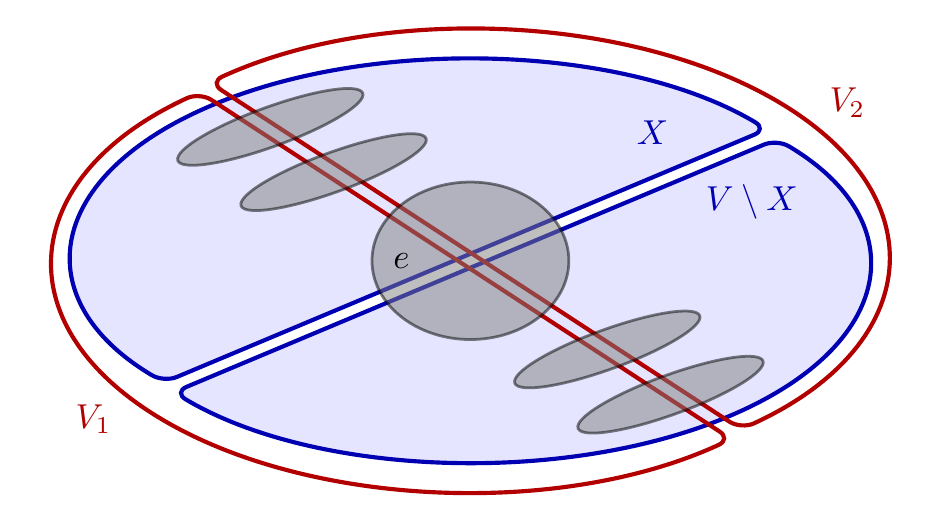}
~
\includegraphics[width=0.48\linewidth]{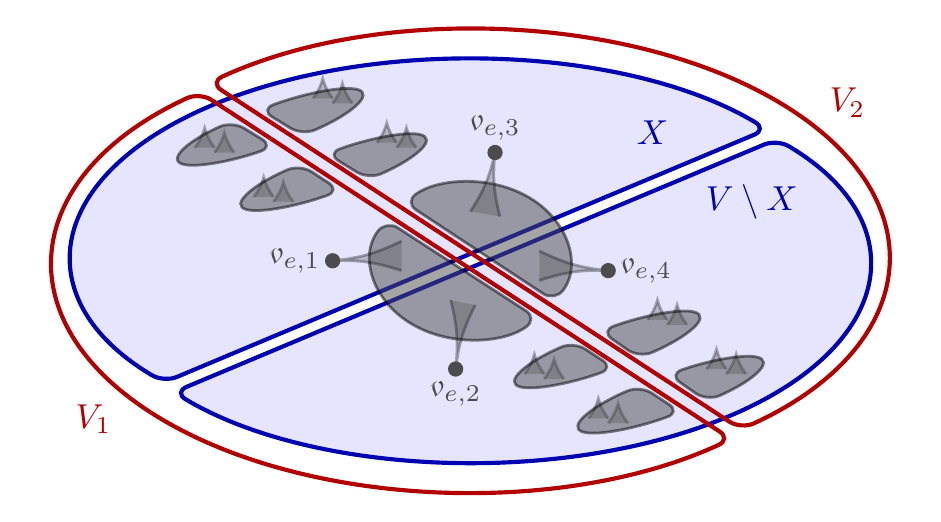}
\caption{An illustration to the proof of \Cref{lem:combine-sparsifiers}.
The gray circles represent hyperedges that cross the bipartition $(V_1, V_2)$ in the divide and conquer framework. When these hyperedges are separated, new anchor vertices are introduced and added to the terminal sets.
The newly created terminal vertices are forced to join different sides of the cut, if and only if the separated hyperedge crosses the $(A, B)$-mincut $(X, V\setminus X)$.}
\label{fig:divide-and-conquer-illustration}
\end{figure}

Suppose $H_1$ is a $(\T_1, c)$-sparsifier of $\hat{G}[V_1]$ and $H_2$ is a $(\T_2, c)$-sparsifier of $\hat{G}[V_2]$ obtained from the conquer step (\Cref{line:conquer-step}).
Let $(Y_1, \pi(V_1\cup A_1)\setminus Y_1)$ and $(Y_2, \pi(V_2\cup A_2)\setminus Y_2)$ be a
$(\pi(\Aextend|_{V_1\cup A_1}), \pi(\Bextend|_{V_1\cup A_1}))$-mincut on $H_1$ and a $(\pi(\Aextend|_{V_2\cup A_2}), \pi(\Bextend|_{V_2\cup A_2}))$-mincut on $H_2$ respectively.
Notice that every hyperedge in $\Sep((\partial X)\cap S, V_1, V_2)$ are in $\partial (Y_1\cup Y_2)$.
Let $Y_0 := Y_1\cup Y_2$ and after removing all anchor vertices we get $Y = Y_0\cap V$.
Now $(Y, \pi(V)\setminus Y)$ is a $(\pi(A), \pi(B))$-cut.
Since for each hyperedge $e\in (\partial_GX)\cap S$, $e$ is separated into two hyperedges and both of them are in $\partial_{H_1\cup H_2}Y_0$, we have
\begin{equation}\label{eqn:partial-y}
|\partial_H Y| \le |\partial_{H_1\cup H_2} Y_0| - |(\partial_G X)\cap S|.
\end{equation}
Notice that the inequality in \Cref{eqn:partial-y} comes from the fact that 
 $\partial_{H_1\cup H_2} Y_0$ may or may not contain more separated hyperedges from $\Sep(S\setminus \partial X, V_1, V_2)$.

Finally we obtain
\begin{align*}
\mincut_{H}^c(\pi(A), \pi(B)) &\le |\partial_H Y| \tag{$Y$ is some $(\pi(A),\pi(B))$-cut on $H$}\\
&\le |\partial_{H_1\cup H_2} Y_0|-|(\partial X)\cap S| \tag{by~\Cref{eqn:partial-y}}\\
&= |\partial_{H_1} Y_1| + |\partial_{H_2} Y_2| - |(\partial X)\cap S| \tag{$Y_0$ is the disjoint union  $Y_1\cup Y_2$}\\
&\le |\partial_{\GSep} X^\mathsf{sep}| - |(\partial X)\cap S| \tag{$X^\mathsf{sep}$ is some $(\Aextend, \Bextend)$-cut}  \\
&= |\partial X|  \tag{exactly $|(\partial X)\cap S|$ hyperedges were separated}\\
&= \mincut_G^c(A, B) \tag{$X$ is an $(A, B)$-mincut}
\end{align*}
as desired.

\noindent\textbf{Part 2.} The proof of $\mincut_H^c(\pi(A), \pi(B))\ge \mincut_G^c(A, B)$ is very similar to Part 1, so we defer the proof (for completeness) in \Cref{appendix:part-2-proof}.
\end{proof}

\subsection{$(5c, c)$-Edge-Unbreakable Terminals}

Let $G=(V, E)$ be a hypergraph and let $\T\subseteq V$ be the set of terminals.
By adopting the notations from \cite{LSS20},
we say that a terminal set $\T$ is \emph{$(5c, c)$-edge-unbreakable} on $G$ if for any bipartition $(V_1, V_2)$ of $V$
with no more than $c$ crossing edges $|E_G(V_1, V_2)|\le c$, either $|\T|_{V_1}| < 5c$ or $|\T|_{V_2}| < 5c$.
That is, if there is a cut of size at most $c$, then at least one of the sides has less than $5c$ induced terminals.

Liu~\cite{Liu20} obtained an $(\T, c)$-sparsifier of size $O(|\T|c^2)$
with a $(5c, c)$-edge-unbreakable terminal set $\T$ 
where each
terminal vertex $v\in \T$ has degree $1$. It turns out that Liu's techniques naturally extend to hypergraphs. We prove the following in \Cref{appendix:proof-of-gammoid}.

\begin{restatable}[]{lemma}{Lemmagammoid}
\label{lem:gammoid}
Let $G=(V, E)$ be a hypergraph and let $\T\subseteq V$ be a set of degree-1 terminals.
If $\T$ is $(5c, c)$-edge-unbreakable on $G$, then there exists a subset $E'\subseteq E$ with $O(|\T | c^2)$ hyperedges,
such that $G/(E-E')$ is a $(\T, c)$-sparsifier of $G$.
\end{restatable}

\section{Existence of $(\T, c)$-Sparsifiers with $O(kc^3)$ Hyperedges}\label{subsection:Existence}

With all the tools equipped in the previous section, 
we are able to prove the existence of a $(\T, c)$-sparsifier with $O(|\T|c^3)$ hyperedges.

\begin{restatable}[]{theorem}{existencesparsifier}
\label{thm:existence-of-sparsifier}
Let $G=(V, E)$ be a hypergraph and $\T\subseteq V$ be the set of terminals.
Then there is a subset $E'\subseteq E$ such that $|E'|=O(|\T|c^3)$ and the contracted hypergraph $G/(E - E')$ is $(\T, c)$-equivalent to $G$.
\end{restatable}

To prove \Cref{thm:existence-of-sparsifier}, it suffices to prove the following \Cref{lem:exist-of-sparsifier-reduced} where every terminal vertex has degree 1:

\begin{restatable}[]{lemma}{existencesparsifierdegone}
\label{lem:exist-of-sparsifier-reduced}
Let $G=(V, E)$ be a hypergraph and $\T\subseteq V$ be the set of degree 1 terminals.
Then there is a subset $E'\subseteq E$ such that $|E'|=O(|\T|c^2)$ and the contracted hypergraph $G/(E - E')$ is $(\T, c)$-equivalent to $G$.
\end{restatable}

\begin{proof}[Proof of \Cref{thm:existence-of-sparsifier}]
Without loss of generality, we may assume that each vertex in $\T$ has degree at most $c$, by duplicating each terminal vertex and add $c$ parallel edges between the duplicated vertex and the original vertex.
Let $\T$ be the terminal set of an input instance.
Now, assuming each terminal has degree at most $c$, we can further duplicate each of these terminals $c$ times so we have a set $\T'$ of at most $|\T|c$ degree-1 terminal vertices.
By \Cref{lem:exist-of-sparsifier-reduced}, there exists a subset $E'\subseteq E$ such that $|E'|=O(|\T'|c^2)=O(|\T|c^3)$ and the contracted hypergraph $G/(E-E')$ is $(\T, c)$-equivalent to $G$.
\end{proof}

To prove \Cref{lem:exist-of-sparsifier-reduced}, we first present an algorithm \algslow (See \Cref{alg:sparsify-slow}).
The algorithm recursively apply divide and conquer framework until the terminal set  is $(5c, c)$-edge-unbreakable as the base case.
After applying \Cref{lem:gammoid} on each base case,
the algorithm combines the sparsifiers from the subproblems using \Cref{lem:combine-sparsifiers}.

\begin{algorithm}[h]
\caption{\algslow$\algslow(G, \T, c)$}\label{alg:sparsify-slow}
\DontPrintSemicolon
\SetAlgoLined
\KwIn{An undirected unweighted multi-hypergraph $G$, a set of degree-1 vertex terminals $\T\subseteq V$, and a constant $c$.}
\KwOut {A $(\T, c)$-sparsifier $H$ for $G$.}
\eIf{$\T$ is $(5c, c)$-edge-unbreakable}{
Construct $H$, a $(\T, c)$-sparsifier of $G$ using \Cref{lem:gammoid}.\label{line:construct-sparsifier-in-subgraph}\\
\Return $H$.
}{
Let $(V_1, V_2)$ be a bipartition of $V(G)$ that refutes the $(5c, c)$-edge-unbreakable property. That is,
$|E_{G}(V_1, V_2)| \le c$ but $|\T\cap V_1|\ge 5c$ and $|\T\cap V_2|\ge 5c$.\label{line:split-condition}\\ 
Obtain $\begin{cases} H_1 \gets \algslow(\hat{G}[V_1], \T_1, c), \text{ and }\\ 
H_2\gets \algslow (\hat{G}[V_2], \T_2, c).\end{cases}$. \label{line:split-by-violating-cut}\\
\Return $H\gets H_1\cup H_2 \cup E(V_1, V_2)-\Sep(S, V_1, V_2)$.\label{line:exists-merge-all-parts}
}
\end{algorithm}

\Cref{lem:algslow-correctness} and \Cref{lem:algslow-size} give the correctness proof and the size to the returned $(\T, c)$-sparsifier from \Cref{alg:sparsify-slow}.

\begin{lemma}\label{lem:algslow-correctness}
\Cref{alg:sparsify-slow} returns a $(\T, c)$-sparsifier of $G$.
\end{lemma}
\begin{proof}
First we notice that all vertices in $\T_1$ and $\T_2$ have degree 1 in $\hat{G}[V_1]$ and $\hat{G}[V_2]$ respectively: the anchor vertices have degree 1 and so the recursive calls in \Cref{line:split-by-violating-cut}  are valid.
The correctness is then recursively guaranteed by \Cref{lem:combine-sparsifiers} (divide-and-conquer step) and \Cref{lem:gammoid} (base case). 
\end{proof}

\begin{restatable}[]{lemma}{Lemmaalgorithmslowsize}
\label{lem:algslow-size}
Let $G$ be a hypergraph,
$\T\subseteq V$ is the set of degree-1 terminal vertices, and let $c$ be a constant.
Let $H=\algslow(G, \T, c)$ be the output of \Cref{alg:sparsify-slow}. Then $H$ has at most $O(|\T|c^2)$ hyperedges.
\end{restatable}

The proof to \Cref{lem:algslow-size} is via  a potential function similarly defined in Liu~\cite{Liu20}.

\begin{proof}
The execution to \Cref{alg:sparsify-slow} defines a recursion tree.
If $|\T|<5c$, then the recursion terminates immediately because $\T$ is trivially $(5c, c)$-edge-unbreakable by definition 
and a $(\T, c)$-sparsifier of $O(|\T|c^2)$
hyperedges is returned by \Cref{lem:gammoid}.
Assume that $|\T| \ge 5c$, then
each recursive call on the subproblem $(G', \T')$ guarantees that $|\T'|\ge 5c$.

Now, it suffices to use the following potential function to prove that the total number of terminal vertices in all recursion tree leaves can be bounded by $O(|\T|)$.
Define a potential function for each subproblem $(G', \T')$ to be $\Phi(G', \T') := |\T'|-5c$.
Then, according to \Cref{line:split-condition}, whenever $(G', \T')$ splits into two subproblems $(\hat{G'}[V_1], \T'_1)$ and $(\hat{G'}[V_2], \T'_2)$ we have 
\begin{align*}
\Phi(\hat{G'}[V_1], \T'_1) + \Phi(\hat{G'}[V_2], \T'_2) &\le 
|\T'\cap V_1| + |\T'\cap V_2| + 4|E_{G'}(V_1, V_2)| - 10c\\
&\le \Phi(G', \T') - c.
\end{align*}

Since every subproblem has a non-negative potential, and the sum of potential decreases by $c$ at each divide-and-conquer step, the total number of leaf cases do not exceed $\Phi(G, \T)/c \le |\T|/c$.
Hence, the total size from the base case is at most $\sum_{(G', \T')\text{: base case}} |\T'|\le \Phi(G, \T) + (5c)(\# \text{ of leaf cases}) = O(|\T|)$.

By \Cref{lem:gammoid}, the total number of hyperedges returned from \Cref{line:construct-sparsifier-in-subgraph} is at most $O(|\T|c^2)$.
The total number of hyperedges added back at \Cref{line:exists-merge-all-parts} is at most the number of divide-and-conquer steps times the cut size, which is at most $|\T|$.
Therefore,
the output $(\T, c)$-sparsifier $H$ has at most $O(|\T|c^2)$ hyperedges as desired.
\end{proof}

\begin{proof}[Proof of \Cref{lem:exist-of-sparsifier-reduced}]
\Cref{lem:exist-of-sparsifier-reduced} follows immediately after the correctness proof (\Cref{lem:algslow-correctness}) and upper bounding the number of hyperedges (\Cref{lem:algslow-size}).
\end{proof}

\section{An Almost-linear-time Algorithm Constructing a Sparsifier} \label{subsection:Algorithm}

This section is devoted to proving part (1) in \Cref{thm:main}.
That is, we give a almost-linear-time (assuming a constant rank) algorithm that constructs a contraction based $(\T, c)$-sparsifier of $O(|\T|c^3)$ hyperedges
which
matches with \Cref{thm:existence-of-sparsifier} up to a constant factor.
We summarize the result in \Cref{conclusion}.

\begin{restatable}[]{theorem}{conclusion}
\label{conclusion}
Let $G=(V, E)$ be a hypergraph with $n$ vertices, $m$ hyperedges, and rank $r=\max_{e\in E}|e|$.
Let $\T\subseteq V$ be a terminal set $\T\subseteq V$.
Then there exists a randomized algorithm which constructs a $(\T, c)$-sparsifier with $O(|\T|c^3)$ hyperedges in $O(p + n(rc\log n)^{O(rc)}\log m)$ time.
\end{restatable}

\subparagraph{Overview of the algorithm.} Although \Cref{alg:sparsify-slow} can construct a $(\T,c)$-sparsifier with $O(|\T|c^3)$ hyperedges, it is slow because we do not have an efficient algorithm searching for a bipartition that violates the $(5c, c)$-edge-unbreakable property. 

To construct our contraction-based $(\T,c)$-sparsifier, all we need to do is identifying \emph{essential} hyperedges and contract non-essential ones. Essential hyperedges are indispensable to maintaining $\mincut$ between terminals (See \Cref{def:essential-edges}).
It seems to be challenging to identify essential hyperedges on an arbitrary graph without a $(5c, c)$-edge-unbreakable guarantee.
Fortunately, we notice there is an efficient way to identify essential hyperedges in an expander.

Naturally, we can utilize \expdec (where the version for hypergraphs is explicitly stated in \cite{LongS22}) which splits a hypergraph into expanders.
Expander decomposition not only guarantees expander sub-hypergraphs, but also fits in the divide-and-conquer framework indicated by \Cref{lem:combine-sparsifiers} with a favorable almost-linear time. Then, we can focus on identifying essential hyperedges in an expander.

To identify essential hyperedges in an expander, we first
enumerate all \emph{connected cuts}\footnote{In Chalermsook et al.~\cite{chalermsook2021vertex}, the concept of connected cuts is not explicitly defined. We give a formal definition in \Cref{def:connected-cuts} and hope it clarifies some ambiguity in their paper.} with value at most $c$ ---
the sub-hypergraph induced  by the smaller side of a connected cut is connected.
Then, we build a \emph{pruned auxiliary graph} based on the cuts we have enumerated. The pruned auxiliary graph leads to an efficient way identifying essential hyperedges. Finally, we contract all detected non-essential hyperedges. We call the above procedure that sparsifies an expander \phisparsify.

With \expdec and \phisparsify procedures introduced above, we are able to construct the $(\T, c)$-sparsifier on general hypergraphs of size $O(|\T|c^3)$ efficiently. 
Our algorithm (\Cref{alg:sparsify-fast}) is based on Chalermsook et al.~\cite{chalermsook2021vertex} and consists of iterations of \expdec and \phisparsify.
Each iteration implements the divide-and-conquer framework shown by \Cref{alg:divide-and-conquer}: we first apply \expdec and decompose the hypergraph into $\phi$-expanders. Then we apply  \phisparsify to sparsify the $\phi$-expanders. Finally, we glue all sparsifiers of the $\phi$-expanders by recovering the inter-cluster hyperedges between the $\phi$-expanders.  
Similar to \cite{chalermsook2021vertex}, we prove that  $O(\log m)$ iterations suffice to obtain a $(\T, c)$-sparsifier of $O(|\T|c^3)$ hyperedges.

\paragraph{Overview of this section.} In \Cref{subsection:Enumeration-of-cuts}, we first describe the settings of the expander decomposition, and then we present an algorithm for enumerating all connected cuts with value at most $c$ in an expander. In \Cref{sec:auxiliary-graph}, we build a pruned auxiliary graph using connected cuts where the algorithm can identify non-essential hyperedges easily.
Then we give an algorithm $\phisparsify$ that produces a sparsifier of an expander with the help of the auxiliary graph.
Finally, \Cref{conclusion} can be directly proved by combining expander decomposition and the $\phisparsify$ algorithm, which is \Cref{alg:sparsify-fast} and it is presented in \Cref{sec:Sparsification-algorithm}.

\subsection{Enumeration of Small Cuts in an Expander} \label{subsection:Enumeration-of-cuts}

To illustrate what an expander is, we first define \emph{conductance}.

\begin{definition}[Conductance of hypergraphs]
\label{def:conductance}

Let $G = (V,E)$ be a hypergraph and a proper subset $S \subsetneq V$. We define the $conductance$ of $S$ to be 
$$\Phi_G(S) = \frac{|\partial{S}|}{\min{(|E(S)|,|E(V\setminus S)|)}}$$
and the define conductance of $G$ to be the minimum conductance over all proper subsets of vertices:
$$\Phi(G) = \min_{S: \emptyset \subsetneq S\subsetneq V}{\Phi_G(S)}$$

\end{definition}

We call a hypergraph with conductance $\phi$ an \emph{$\phi$-expander}.
Long and Saranurak~\cite{LongS22} give an almost-linear-time algorithm \expdec that partitions a hypergraph into $\phi$-expanders. 

\begin{lemma}[\cite{LongS22}]
\label{lem:expander-decompose-for-hypergraphs}
There exists a randomized algorithm \expdec that, given any unweighted hypergraph $G=(V,E)$ with $n$ vertices, $m$ hyperedges, and total size $p= \sum_{e\in E}|e|$, and any parameter $\phi>0$, with high probability computes in $p^{1+o(1)} $ time a partition $\{V_1,\dots,V_k\}$ of $V$ such that
\begin{itemize}
    \item for all $i$, $\Phi(G[V_i]) \ge \phi$, and
    \item the number of crossing hyperedges is at most $\phi m\, \polylog(n)$ (we say that an edge $e$ is crossing if $e$ contains vertices from at least two parts $V_i$ and $V_j$ for some $i\neq j$).
\end{itemize}
\end{lemma}

We also note that any cut $(X, V \setminus X)$ of value at most $c$ in a $\phi$-expander, according to \Cref{def:conductance}, satisfies 
\begin{equation}\label{eq:hyperedge-upper-bound}
 \min\{|E(X)|, |E(V \setminus X)|\} \leq c \phi^{-1}.  
\end{equation}

For now, we will focus on constructing $(\T, c)$-sparsifiers on $\phi$-expanders. 
This is because \expdec can be easily incorporated into the divide-and-conquer framework (\Cref{alg:divide-and-conquer}) as will be formally shown in \Cref{alg:sparsify-fast} near the end of the section.

Given a $\phi$-expander $G$ with terminal set $\T$, our goal is to identify an \emph{essential}  hyperege $e$, that is, to check whether $e$ is in \emph{every} $(A,B)$-$\mincut$s with value at most $c$ for some disjoint terminal sets $A$ and $B$.
In a general hypergraph there could be as many as $O(2^n)$ mincuts with value at most $c$ to check but in a $\phi$-expander there is much less. 
It turns out that finding all \emph{connected cuts} with value at most $c$ suffices to identify essential hyperedges:

\begin{definition}[Connected Cuts]
\label{def:connected-cuts}
For a hypergraph $G = (V,E)$, let $(X,V\setminus X)$ be a cut where $X \subseteq V$ and $|E(X)| \leq |E(V\setminus X)|$. We say $(X,V\setminus X)$ is a \emph{connected cut} if and only if $G[X]$ is connected. 
\end{definition}

Suppose there is a connected cut $(X, V\setminus X)$ with value at most $c$,
then by the property of the $\phi$-expander we know that $|E(X)|\le c\phi^{-1}$ (\Cref{eq:hyperedge-upper-bound}).
Using the assumption that $G[X]$ is connected, 
we can reach any boundary hyperedge in $\partial X$ by invoking a DFS traversal from a vertex $v_{\text{seed}}\in X$ within $c\phi^{-1}$ steps.
Since there are at most $c$ boundary hyperedges, their sizes add up to at most $rc$.
Then, we recursively ``guess and trim'' these boundary hyperedges for at most $rc$ times on $G$, obtaining $G[X]$ at the end.
\enumcuts (\Cref{alg:enumerate-cuts}) tries every possible $v_{\text{seed}}$ and invokes the helper function \enumcutshelp (\Cref{alg:enumerate-cuts-help}) that performs the ``guess and trim'' procedure.
We summarize the guarantee of \enumcuts in \Cref{lem:enumerate-cuts}:

\begin{algorithm}[h]
\caption{\enumcuts$(G,\phi,r,c)$}\label{alg:enumerate-cuts}
\SetAlgoLined
\KwIn{A $\phi$-expander hypergraph $G = (V,E)$ with $\mathrm{rank}$ $r$, and a threshold parameter $c$.}  
\KwOut{All connected cuts with value at most $c$.} 
$\C \gets \emptyset$.  \tcp{Stores all found connected cuts.}
\For{each $v_{\text{seed}} \in V$ \label{line:loop-over-vertices}}
{
    \tcc{Invokes a helper function to find all connected cuts involving $v_{\text{seed}}$.}
    $\C \gets \C \cup \enumcutshelp(0, G, G, \phi, r, c, v_{\text{seed}})$. \tcp*{See \Cref{alg:enumerate-cuts-help}.}
}%
\Return $\C$.\label{line:alg2-return}%
\end{algorithm}

\newcommand{\Vmark}{V_{\sf{marked}}}

\begin{algorithm}[h]
\caption{\enumcutshelp$(\mathit{depth}, H, G, \phi, r, c, v_{\rm seed})$}\label{alg:enumerate-cuts-help}
\SetAlgoLined

\KwIn{The current recursion depth $\mathit{depth}$. A hypergraph $H=(V, E)$ with rank $r$. The original hypergraph $G$. Parameters $c$ and $\phi$. A seed vertex $v_{\rm seed}\in V$.}
\KwOut{All connected cut with value at most $c$ so that $v_{seed}$ is in the smaller side.} 
\eIf{$\mathit{depth}\le rc$}
    {Run DFS from $v_{\rm seed}$ on $H$ and {\textbf{stop}} as soon as visiting $c\phi^{-1}+1$ hyperedges.\label{line:run-dfs}\\
    Let $\hat{E}$ be the set of visited hyperedges and $X$ be the set of visited vertices.\\
    \eIf{DFS gets stuck before visiting $c\phi^{-1}+1$ hyperedges\label{line:DFS-get-stuck}}
    {
        \eIf{$|\partial_G X| \leq c $ \label{line:size-threshold}}
        {\Return $\{(X, V\setminus X)\}$. \tcc{Some connected cut with value at most $c$ is found.}}
        {\Return $\emptyset$.}
    }{
    $\S\gets \emptyset$.\\
    \For{each $e\in \hat{E}$ and for each $v\in e$, $v\neq v_{\rm seed}$}
    {
    \label{line:loop-over-visited-vertices}%
    Let $e'\gets e\setminus v$.\tcc{modify the boundary hyperedge into a smaller one.}
    \tcp{A recursive call with $v$ being removed from $e$. \label{line:recursive-call}}
    $\S \gets \S \cup \enumcutshelp(\mathit{depth}+1, H-e+e', G, \phi, r, c, v_{\rm seed})$
    }
    \Return $\S$.
    }
}
{ %
    \Return $\emptyset$.
}

\end{algorithm}

\begin{restatable}[]{lemma}{Lemmaenumeratecuts}
\label{lem:enumerate-cuts}
Given a $\phi$-expander hypergraph $G=(V, E)$, there are at most $|V|(rc\phi^{-1})^{rc}$ connected cuts with value at most $c$. Moreover, \Cref{alg:enumerate-cuts} enumerates all of them in $O(|V|(rc \phi^{-1})^{rc+1})$ time.
\end{restatable}

\begin{proof}
We will first show that (1) all connected cuts with value at most $c$ can be found by \Cref{alg:enumerate-cuts} and that (2) all the returned cuts from \Cref{alg:enumerate-cuts} are connected cuts with value at most $c$.
Then, we show that the running time of \Cref{alg:enumerate-cuts} is $O(n(rc\phi^{-1})^{rc+1})$ in part (3).

\subparagraph{Part (1).}
We first show that all connected cuts with value at most $c$ can be found by \Cref{alg:enumerate-cuts}.
For an arbitrary connected cut $(X,V\setminus X)$ with value at most $c$ containing some vertex $v_{seed}$ in $G = (V,E)$ where $|E(X)| \leq |E(V \setminus X)|$, we have $|\partial X| \leq c$ and $G[X]$ is connected.
According to \Cref{eq:hyperedge-upper-bound}, we have $|E(X)| \leq c\phi^{-1}$.
So, if a DFS starting from $v_{seed} \in X$ yields after visiting $c\phi^{-1}+1$ hyperedges, then there must exist a hyperedge $e_{cut}$ among the visited hyperedges such that $e_{cut} \in \partial X$, which implies $e_{cut} \setminus X \not = \emptyset$.
Thus, \Cref{alg:enumerate-cuts-help} removes some $v \in e_{cut} \setminus X$ from $e_{cut}$ in \Cref{line:recursive-call}. %
The algorithm recursively removes vertices in $V \setminus X$ from hyperedges in $\partial X$. As a result, all vertices in $V\setminus X$ and the boundary hyperedges
will be removed from all boundary hyperedges at some point in the recursion. That means the boundary hyperedges only contains vertices in $X$ and the DFS will get stuck before visiting $c\phi^{-1} + 1$ hyperedges.
Notice that the total number of vertices in boundary edges that gets removed is less than $r|\partial X|\le rc$.
In addition, since $G[X]$ is connected, the DFS gets stuck
after visiting all vertices in $X$ and then 
\Cref{alg:enumerate-cuts-help} returns the cut $(X,V\setminus X)$.
\Cref{alg:enumerate-cuts} iterates over all vertices in $V$ as a seed vertex, therefore, 
it returns all connected cuts with value at most $c$ in  $G$. 

\subparagraph{Part (2).} We show that all the cuts that \Cref{alg:enumerate-cuts} finds are connected cuts with value at most $c$.
The value of all the cuts the algorithm finds is less than $c$ according to \Cref{line:size-threshold} in \Cref{alg:enumerate-cuts-help}.
Now, we observe that
DFS only probes adjacent vertices. When DFS gets stuck, the cut $(X, V\setminus X)$ that DFS gives is a connected cut in the modified hypergraph.
Since the modified hypergraph is obtained only via removing vertices from hyperedges,
it is not hard to see that $G[X]$ is connected.

\subparagraph{Part (3).}
\Cref{alg:enumerate-cuts-help} is invoked $|V|$ times.
The runtime per execution of \Cref{alg:enumerate-cuts-help} can be upper bounded by the size of the recursion tree multiplied with the worst case time needed for DFS (\Cref{line:run-dfs}).
From the for-loop in \Cref{line:loop-over-visited-vertices}, each node in the recursion tree has at most $rc\phi^{-1}$ children and the depth of the recursion tree is bounded by $rc$.
So, the size of the recursion tree is at most $(rc\phi^{-1})^{rc}$.
Moreover, each DFS (\Cref{line:run-dfs}) takes $O(rc\phi^{-1})$ time since at most $c\phi^{-1}+1$ hyperedges with rank $r$ are visited.
Thus, the runtime can be bounded by $O(|V|(rc\phi^{-1})^{rc+1})$.
\end{proof}

\subsection{Sparsification via an Auxiliary Graph} \label{sec:auxiliary-graph}

We now introduce the algorithm \phisparsify which constructs a $(\T,c)$-sparsifier with  $O(|\T|c^{3})$ hyperedges from a $\phi$-expander. 
Our algorithm \phisparsify detects the hyperedges that are not \emph{essential} to some $(\T, c)$-sparsifier and contracts them. The formal definition of essential hyperedges is as follows.

\begin{definition}[Essential and Non-essential Hyperedges]
\label{def:essential-edges}
A hyperedge $e$ is said to be \emph{essential} if there exists a partition of terminals $(A, \T\setminus A)$ such that all $(A, \T\setminus A)$-mincuts with value at most $c$ contain $e$. Otherwise, the hyperedge is non-essential.
\end{definition}

Notice that all essential hyperedges cannot be contracted or removed as they will affect $(A,\T\setminus A)$-mincut value for some $A \subset \T$.
Hence, any $(\T, c)$-sparsifier of $G$ must include all essential hyperedges in $G$. That said, by the existence theorem (\Cref{thm:existence-of-sparsifier}), the number of essential hyperedges on $G$ is at most $O(|\T|c^3)$.

\subsubsection{Auxiliary Graph (and its Subtle Issue)}

According to \Cref{def:essential-edges}, a hyperedge $e$ is non-essential if
for each partition of terminals $(A, \T\setminus A)$, there exists a $(A, \T\setminus A)$-mincut with value at most $c$ that does not contain $e$ in the boundary. 
A straightforward way to check whether a hyperedge is non-essential,
is to check whether the value of a $(A, \T\setminus A)$-mincut for some $A \subset \T$ would be affected after the hyperedge is removed.

Chalermsook et al.~\cite{chalermsook2021vertex} 
utilize \enumcuts (the normal graph version) to check whether an edge is non-essential.
In particular, they
introduced the notion of \emph{auxiliary graph} on a $\phi$-expander that helps identifying non-essential edges.
Unfortunately, their definition (and the construction) to the auxiliary graph does not have the desired property that leads to a sufficient criterion recognizing an essential edge.
To illustrate this, we first formally state 
the definition of an auxiliary graph from Chalermsook et al., which naturally extends to hypergraphs.

\begin{definition}[Auxiliary Graph]
\label{def:auxiliary-graph}
Let hypergraph $G=(V, E)$ be a $\phi$-expander.
An \emph{auxiliary graph} $\Gaux=(\Vaux, \Eaux)$ for $G$ has its vertex set $\Vaux$ consisting of three disjoint parts $\Vaux= P_0 \cup C_0 \cup E_0$.
Each edge in $\Eaux$ either connects  the elements between $P_0$ and $C_0$, or connects the elements between $C_0$ and $E_0$. The elements in the three parts and the edges between them are defined as follows:
\begin{itemize}[itemsep=0pt]
\item $C_0$ contains the cuts returned by \Cref{alg:enumerate-cuts} with value at most $c$ on $G$.
\item Each cut $c = (X, V\setminus X)\in C_0$ induces a terminal partition $p=(\T\cap X, \T\setminus X)$. Let $P_0$ be the collection of all such terminal partitions (excluding the trivial partitions where $\T\cap X=\emptyset$ or $\T\subseteq X$). Moreover, we add $(p, c)$ to $\Eaux$ if $c$ is a $p$-mincut.
\item $E_0=E$. For each $c=(X, V\setminus X)\in C_0$, we add $(c, e)$ to $\Eaux$ for all $e\in \partial X$.
\end{itemize}
\end{definition}

\paragraph{A Subtle Issue in \cite{chalermsook2021vertex}.}
Chalermsook et al.~\cite{chalermsook2021vertex} proposed the algorithm that recognizes an edge $e\in E_0$ to be essential iff there is a partition $p\in P_0$ such that $N(p)\subseteq N(e)$.
That is, all $p$-mincuts in $C_0$ neighboring to $p$ contains $e$. Unfortunately, this algorithm has a subtle issue.
Indeed, if $e$ is an essential edge then there exists such a partition. However, it is not the case conversely.
The statement would have been true if $C_0$ is the set of \emph{all} cuts with value at most $c$.
However, the cuts returned by \Cref{alg:enumerate-cuts} are connected cuts only. Only considering connected cuts will result in recognizing non-essential edges to be essential (even in normal graphs). We give an example below.

\paragraph{An Example.}
Consider the example shown in \Cref{fig:aux-graph-counterexample}.
The graph has $7$ vertices $v_1, v_2, \ldots, v_7$ and $3$ of them are terminal vertices $\T=\{v_1, v_2, v_7\}$. 
There are $3$ proper partitions of $\T$,
which separates each of $v_1$, $v_2$, and $v_7$ from the rest of terminals respectively.
For the terminal partitions $p_1:=(\{v_1\}, \{v_2, v_7\})$ and $p_2:=(\{v_2\}, \{v_1, v_7\})$, there are a unique $p_1$-mincut $(\{v_1\}, V\setminus \{v_1\})$ and a unique $p_2$-mincut $(\{v_2\}, V\setminus \{v_2\})$. Therefore, the cutting edges $a$ and $b$ respectively are essential by \Cref{def:essential-edges}.
For the terminal partition $p_7:=(\{v_7\}, \{v_1, v_2\})$, there are two $p_7$-mincuts: $(\{v_1, v_2\},\{v_3, v_4, v_5, v_6, v_7\})$ and $(\{v_1, v_2, v_3, v_4\},\{v_5, v_6, v_7\})$. The edges crossing these mincuts are $\{a, b\}$ and $\{d, e\}$ respectively.
So, by \Cref{def:essential-edges}, edges $d$ and $e$ are non-essential since none of the terminal partition has $d$ or $e$ appearing in every mincut.

However, \Cref{alg:enumerate-cuts} does not report one of the $p_7$-mincut $(\{v_1, v_2\}, \{v_3, v_4, v_5, v_6, v_7\})$ since this cut is not a connected cut.
Therefore, \Cref{alg:enumerate-cuts} returns only one $p_7$-mincut and 
$d$ and $e$ will be erroneously recognized as essential edges in the auxiliary graph.

\begin{figure}[h]
    \centering
    \includegraphics[width=5.5cm]{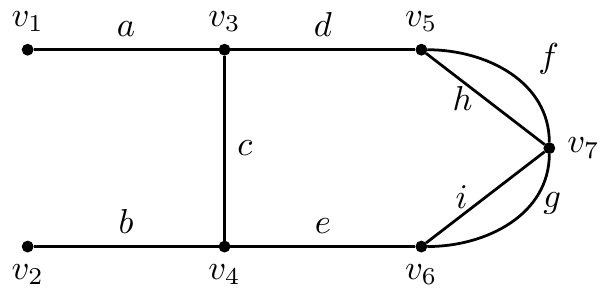}\hspace*{1.5cm}
    \includegraphics[width=4.5cm]{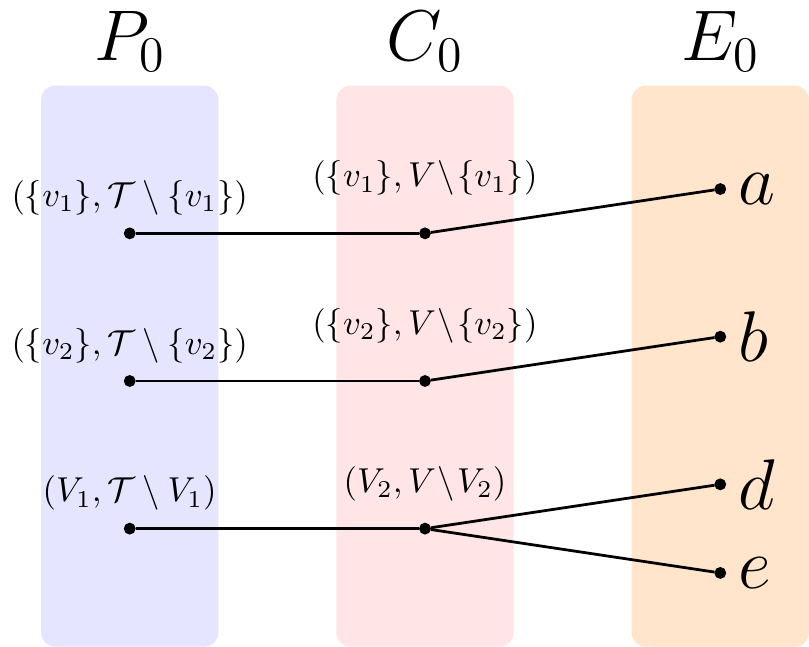}
    \caption{A graph $G$ and its auxiliary graph ($\T = \{v_1, v_2, v_7\}, V_1 = \{v_1, v_2\}, V_2 = \{v_1, v_2, v_3, v_4\}$). Notice that in the auxiliary graph, edges $d$ and $e$ are incorrectly considered as essential.}
    \label{fig:aux-graph-counterexample}
\end{figure}

To resolve the problem of incorrectly recognizing unessential hyperedges to be essential, we \emph{prune} the auxiliary graph.
The intuition is to prune some partitions as well as their incident nodes and edges in the auxiliary graph so that the unessential edges are removed.
For instance, the terminal partition $(\{v_1, v_2\}, \{v_7\})$ is one of the partitions that should be pruned with its incident nodes and hyperedges from the auxiliary graph. The partitions that should be preserved are called \emph{useful} terminal partitions (For instance, $(\{v_1\}, \{v_2,v_7 \})$ and $(\{v_2\}, \{v_1,v_7 \})$), which will be elaborated on in the next section.

\subsubsection{Pruned Auxiliary Graph}

Let us now formally define \emph{useful} terminal partitions and describe the \emph{pruned auxiliary graph} that helps us to correctly identify all essential hyperedges.
This subsection contains one of the main technical contributions of this paper.

\begin{definition}[Useful Terminal Partition]
\label{def:useful-partition}
A terminal partition of terminals $(A, B)$ is \emph{useful} if $\mincut_G(A, B)\le c$ and for all $(A, B)$-mincut $(X, V \setminus X)$, $G[X]$ is connected.
\end{definition}

With \Cref{def:useful-partition}, we can prune unnecessary partitions from the auxiliary graph $\Gaux$: we only keep the partitions in $P_0$ that are useful, and discard all the related mincuts from $C_0$ and the isolated elements in $E_0$.

Henceforth, we refer $\Gaux$ to be the \emph{pruned auxiliary graph} and $P_0, C_0, E_0$ are pruned as well.
The following lemma shows that, by working on the pruned auxiliary graph, the characterization of all essential hyperedges on $G$ proposed in \cite{chalermsook2021vertex} is now correct.

\begin{restatable}[]{lemma}{Lemmaessentialauxiliary}
\label{lem:essential-auxiliary}
A hyperedge $e\in E$ is essential if and only if $e\in E_0$ and there exists useful $p\in P_0$ such that $N(p)\subseteq N(e)$.
\end{restatable}

\begin{proof}
The backward direction ($\Leftarrow$) follows from \Cref{def:essential-edges}, now we show the forward direction.

\paragraph{($\Rightarrow$):}
Suppose $e \in E$ is an essential hyperedge. Then by \Cref{def:essential-edges}, we know that there are at least one partition $p=(A, B)$ such that all $c$-thresholded $(A, B)$-mincuts involve $e$. If $(A, B)$ is useful then we are done. 
Now, assume that $p$ is not useful, 
then there exists at least one $(A, B)$-mincut $(X, V\setminus X)$ such that $G[X]$ is not connected.
Without loss of generality we may further assume that $|A|$ is the smallest among all non-useful partitions $(A', B')$ with $\mincut_G(A', B')\le c$ such that $e$ is involved in every $(A', B')$-mincut.

Since $G[X]$ is not connected, $G[X]$ has at least two connected components.
Therefore, there exists a partition on $X=X_1\cup X_2$ such that $E(X_1)\cap E(X_2)=\emptyset$ and, without loss of generality, assume $e\in E(X_1, V\setminus X_1)$.
Notice that $A\cap X_1\neq \emptyset$ (otherwise $(X_2, V\setminus X_2)$ is also a $(A, B)$-mincut but does not involve $e$,) and for the same reason $A\cap X_2\neq \emptyset$.
Now, the new terminal partition $p_{new}=(A\cap X_1, B\cup (A\setminus X_1))$ has the property that every $p_{new}$-mincut $(X_{new}, V\setminus X_{new})$ involves $e$ --- otherwise $(X_{new}\cup X_2, V\setminus X_{new}\setminus X_2)$  would be a $(A, B)$-mincut that does not involve $e$.
Finally, using $|A\cap X_1| < |A|$ and the assumption to minimality of $|A|$, we know that $p_{new}$ is useful so the statement holds.
\end{proof}

\subparagraph{Pruning the Auxiliary Graph.}

Here, we show how to efficiently test whether a terminal partition $(A,\T\setminus A)$ is useful. All we need to do is to find the \emph{\tmcut} $(X^*,V\setminus X^*)$ and check whether $G[X^*]$ is connected. To illustrate the idea, we first define the \amcut. Then we will show that we can check whether $(A,\T\setminus A)$ is useful by checking whether $G[X^*]$ is connected (\Cref{lem:use-amcut-to-identify-useful-partition}). Lastly, to algorithmically find the \tmcut in the hypergrpah, we resort to an algorithm called \locmincut in \cite{forster2020computing}.

\begin{definition}[$A$-minimal $(A,B)$-mincut]
\label{def:closest-mincut}
Let $G = (V,E)$ be a hypergraph and $A,B\subseteq V$ be disjoint subsets. We say an $(A,B)$-mincut $(X^*,V\setminus X^*)$ is the $A$-minimal $(A,B)$-mincut, if, for any other $(A,B)$-mincut $(X,V\setminus X)$, we have $X^* \subseteq X$. 
\end{definition}

For any $A,B$, an $A$-minimal $(A,B)$-mincut is unique. This standard fact follows from submodularity of cut functions on hypergraphs. For completeness, we give the proof in \cref{sec:uniqueness-of-A-minimal-cut}.

\begin{restatable}[]{proposition}{Lemmauniqueminimalcut}
\label{lem:uniqueness-of-A-minimal-mincut}
Let $G = (V,E)$ be a hypergraph. For two disjoint subsets $A,B \subseteq V$, the $A$-minimal $(A,B)$-mincut always exists and is unique.
\end{restatable}

Since the \amcut always exists and is unique, we can check the \tmcut to see whether the terminal partition $(A,\T\setminus A)$ is useful via the following lemma.

\begin{restatable}[]{lemma}{Lemmaidentifyusefulpartition}
\label{lem:use-amcut-to-identify-useful-partition}
Let $H = (V,E)$ be a connected hypergraph and $\T \subset V$ be a terminal set, a terminal partition $(A, \T\setminus A)$ is useful, if and only if the \tmcut $(X^*,V \setminus X^*)$ is connected and $|\partial X^*| \leq c$.
\end{restatable}

\begin{proof}
One direction is trivial: The \tmcut is one of \tpar-mincuts, so it is connected and its value is at most $c$. It remains to prove another direction.

Suppose \tmcut $(X^*,V \setminus X^*)$ is connected and $|\partial X^*| \leq c$. As clearly $\mincut_G(A,\T \setminus A) \le c$, we need to show that all \tpar-mincuts are connected.
Assume for the sake of contradiction that there exists a disconnected \tpar-mincut, which consists of $(X_1, V\setminus (X_1 \cup X_2))$ and $(X_2, V\setminus (X_1 \cup X_2))$ where $E(X_1, X_2) = \emptyset$. We claim that that both $X^* \cap X_1$ and $X^* \cap X_2$ are nonempty.
If the claim holds, we are done because  $G[X^*]$ is connected, so $E(X_1, X_2) \not = \emptyset$, contradicting with $E(X_1, X_2) = \emptyset$. Thus, all \tpar-mincuts are connected.

Finally, to see the claim, assume for contradiction without loss of generality that $X^* \cap X_1 = \emptyset$. Notice that by \Cref{def:closest-mincut}, $A \subseteq X^* \subseteq X_1 \cup X_2$ and thus $A \subseteq X^* \subseteq X_2$. 
So $(X_2,V\setminus X_2)$ is a $(A,\T \setminus A)$-cut of value $|\partial X_2|$. But the $(A,\T \setminus A)$-mincut $(X^* , V\setminus X^*)$ has greater value of $|\partial X_1| + |\partial X_2| > |\partial X_2|$ (where $|\partial X_1|>0$ because the hypergraph $H$ is connected). This is a contradiction.
\end{proof}

\Cref{lem:use-amcut-to-identify-useful-partition} reduces the problem of testing whether a terminal partition $(A,\T\setminus T)$ is useful to checking whether the $A$-minimal \tpar-cut is connected. A randomized algorithm called \locmincut can efficiently find the \tmcut. More details about \Cref{thm:local-min-cut} can be found in \Cref{sec:local-min-cut}.

\begin{restatable}[]{theorem}{ThmLocalMinCut}
\label{thm:local-min-cut}
Suppose that the $A$-minimal $(A,B)$-mincut $(S,V\setminus S)$ has size at most $c$ and volume $\vol(S)=\sum_{v \in S}\deg(v)\le \nu$. 
Then, there is an algorithm $\locmincut(A,\nu,c)$ that finds $S$ in $O(\nu c^3 \log n)$ time with high probability. 
\end{restatable}

To analyze the runtime of constructing a pruned auxiliary graph for an expander, we first discuss about the runtime of pruning a given auxiliary graph.

\subparagraph{Runtime of pruning auxiliary graph.}
To prune an auxiliary graph, for all terminal partitions $(A,\T\setminus A)$, we use \locmincut to find the \tmcut and check whether it is connected. There are at most $(rc \phi^{-1})^{rc}$ terminal partitions in an auxiliary graph.
Let $(S, V\setminus S)$ be the \tmcut.
Because $\vol(S) = \sum_{v\in S} \deg(v) \leq |E(S)|\max_{e\in E}{|e|} = rc\phi^{-1}$,
by setting $\nu = rc\phi^{-1}$ and applying \Cref{thm:local-min-cut},
each terminal partition
needs $O(rc\phi^{-1}c^3\log{n})$ time to check.
Therefore,  the total time of pruning the auxiliary graph is $O((rc\phi^{-1})^{rc+1}c^3\log{n})$.

We now analyze the runtime of constructing a pruned auxiliary graph. 

\begin{restatable}[]{lemma}{LemmaConstructPrunedAuxiliaryGraph}
\label{lem:construct-pruned-auxiliary-graph}
Given a $\phi$-expander $G$, there is an algorithm that constructs a pruned auxiliary graph in $n(rc\phi^{-1})^{O(rc)}$ time.
\end{restatable}

\begin{proof}
\enumcuts runs in $O(n(rc\phi^{-1})^{rc+1})$ time.
It produces at most $n(rc\phi^{-1})^{rc}$ of connected cuts, 
so the construction of (unpruned) auxiliary graph takes $O(nc(rc\phi^{-1})^{rc}) = O(n(rc\phi^{-1})^{rc+1})$ time.
With the discussion above, we know that pruning the auxiliary graph takes additional $O((rc\phi^{-1})^{rc+1}c^3\log n)$ time.
Hence, the total runtime is $O((n+c^3\log n)(rc\phi^{-1})^{rc+1}) = n(rc\phi^{-1})^{O(rc)}$ time.
\end{proof}

An interesting property about usefulness of a partition is that this definition is robust against contraction of any non-essential hyperedge, which is useful in the next section.
\begin{restatable}[]{lemma}{Lemmausefulnesscontraction}
\label{lem:usefulness-contraction}
Let $G=(V, E)$ be a hypergraph and let $\T\subseteq V\cup E$ be the set of terminals.
Let $(A, B)$ be a terminal partition of $\T$
so that there exists a $(A, B)$-mincut of size at most $c$.
Consider any non-essential hyperedge $e\in E$.
Then, $(A, B)$ is useful on $G$ if and only if $(A/e, B/e)$ is useful on $G/e$.
\end{restatable}

\begin{proof}
On one hand, in a connected component of $G$, contracting a hyperedge will results in all of the vertices in the hyperedge being contracted to a single vertex, and the vertex will be in all of the hyperedges that share at least one vertex with the contracted hyperedge. Therefore, the component remains connected. On the other hand, if we uncontract a vertex in a connected component (which means the vertex is connected to all other vertices in the component) the hyperedge, the all of the newly introduced vertices will also be connected to all the other vertices through the new hyperedge resulted from the uncontraction. To conclude, contracting and uncontracting hyperedges will not result in changes in connectedness in the hypergraph, we only need to show that if $e$ is an non-essential hyperedge, then (We only need to show one direction since the lemma states that $\mincut_G(A,B) \leq c$)
\[\mincut_G(A, B) \leq c \Rightarrow \mincut_{G/e}(A/e, B/e) \leq c\]
Since $e$ is non-essential we just fix an $(A,B)$-mincut that does not contain $e$, this cut will remain have cut value at most $c$ after contracting $e$ thus $\mincut_{G/e}(A/e, B/e) \leq c$.
\end{proof}

\subsubsection{The \phisparsify Procedure}

Once we obtain the pruned auxiliary graph of a $\phi$-expander $G$, a simple greedy \phisparsify procedure (\Cref{alg:phi-sparsify}) can be applied for constructing a $(\T, c)$-sparsifier $H$ of $G$.

\begin{algorithm}[h]
\DontPrintSemicolon
\caption{\phisparsify$(G, \T, \phi, r, c)$}\label{alg:phi-sparsify}
\label{phi-sparsify}%
\SetAlgoLined
\KwIn{$\phi$-expander hypergraph $G = (V,E)$ with $\mathrm{rank}$ $r$, terminal set $\T$, threshold parameter $c$.}
\KwOut{A $(\T,c)$-sparsifier of $G$.}
Run $\enumcuts$ and construct the pruned auxiliary graph $\Gaux=(\Vaux,\Eaux)$, where $\Vaux=P_0\cup C_0\cup E_0$.\label{line:phi-sparsify-contruct-auxiliary-graph}\\
Let $E'\gets E\setminus E_0$.\\
\For{each $e \in E_0$ (in any order)}{\label{line:enum-edge}
    \tcc{Use \Cref{lem:essential-auxiliary} to check and identify non-essential edges.}
    Compute the set of partitions $P'_e := P_0\cap N(N(e))$ who has at least one mincut that contains the edge $e$.\label{line:phi-sparsify-find-all-partition}\\
    \If{$\forall$ $p \in P'_e$, $N(p)\not\subseteq N(e)$\label{line:phi-sparsify-check-partition}}{
        Remove $N(e)$ and all incident edges from $\Gaux$ and then Remove all independent vertices from $\Gaux$.
        \label{line:phi-sparsify-simulate-contraction}\\
        
        $E'\gets E'\cup \{e\}$. 
        \label{line:phi-sparsify-simulate-contraction-1}\tcp*{$e$ is non-essential.}
    }
}
\Return $G/E'$.\label{line:return-contracted-graph}
\end{algorithm}

To analyze \Cref{alg:phi-sparsify}, we first introduce \Cref{lem:contraction-invariant}, \Cref{cor:batch-contraction-invariant}, and \Cref{lem:essential-edges-remain-essential}. Finally, \Cref{lem:sparsifier-validity-and-time} summarizes \Cref{alg:phi-sparsify} based on the previous lemmas and corollary.

\Cref{lem:contraction-invariant} guarantees that after each update of the auxiliary graph, the hypergraph represented by the updated auxiliary graph is still viable for the \phisparsify algorithm. 

\begin{restatable}[]{lemma}{Lemmacontractioninvariant}
\label{lem:contraction-invariant}
Let $e$ be the first edge that triggers \Cref{line:phi-sparsify-simulate-contraction} of \Cref{alg:phi-sparsify}.
Suppose that
$\Gaux_1$ is the updated auxiliary graph with vertex set $V_{1}^{\aux} = P_1 \cup C_1 \cup E_1$, 
and let $G_1=G/e$. 
Then, $\Gaux_1$ is exactly the pruned auxiliary graph of $G_1$.
\end{restatable}

\begin{proof}
Let $\Gaux_*$ be the pruned auxiliary graph of $G_1$ with the vertex set $V_{*}^{\aux} = P_* \cup C_* \cup E_*$. To prove $\Gaux_* = \Gaux_1$, it suffices to prove $C_* = C_1$ as the pruned auxiliary graph is uniquely defined (via \Cref{def:auxiliary-graph} and \Cref{lem:usefulness-contraction}) by the connected cuts of value at most $c$.

Notice that given an arbitrary cut $(X, V\setminus X)$ in a hypergraph, contracting a \emph{non-cutting} hyperedge from $E \setminus E(X, V\setminus X)$ does not invalidate the cut nor change the cut value.
Therefore, the contraction of a hyperedge $e$ invalidates every cut that involves $e$.
Moreover, no new connected cut of value at most $c$ joins $C_*$ because contracting a hyperedge in a connected cut does not decrease the cut value.
Hence, $C_*$ contains all the cuts in $C_0$ except the cuts involving $e$. On the other hand, $C_1 = C_0 \setminus N(e)$, which exactly contains all the cuts in $C_0$ except the cuts that involve $e$ as well.
\end{proof}

\begin{corollary}
\label{cor:batch-contraction-invariant}
Right before returning $G/E'$ from \Cref{alg:phi-sparsify} at \Cref{line:return-contracted-graph}, the modified $\Gaux$ is the pruned auxiliary graph of $G/E'$.
\end{corollary}

\begin{proof}
The statement is true by applying \Cref{lem:contraction-invariant} iteratively whenever the algorithm adds an hyperedge to $E'$ and modifies $\Gaux$.
\end{proof}

\begin{lemma}
\label{lem:essential-edges-remain-essential}
Essential hyperedges remains essential under contractions of non-essential hyperedges.
\end{lemma}

\begin{proof}
For an arbitrary essential hyperedge $e$, there is a terminal partition $p$ such that all the $p$-mincuts contain $e$. After contractions of some non-essential hyperedges, all the $p$-mincuts still contain $e$, because contractions never introduce new cuts to a hypergraph.
\end{proof}

With \Cref{lem:contraction-invariant}, \Cref{cor:batch-contraction-invariant}, and \Cref{lem:essential-edges-remain-essential}, we are able to prove the following lemma which summarizes \phisparsify.

\begin{restatable}[]{lemma}{Lemmasparsifiervalidityandtime}
\label{lem:sparsifier-validity-and-time}
Let $G$ be a $\phi$-expander with $n$ vertices and $\T$ be a terminal set.
Then, \Cref{alg:phi-sparsify} produces a $(\T, c)$-sparsifier with $O(|\T|c^3)$ hyperedges and runs in $O(p+ n(rc\phi^{-1})^{O(rc)})$ time, where $p=\sum_{e\in E}|e|$ is the total size of the hypergraph $G$.
\end{restatable}

\begin{proof}
We show correctness of \phisparsify, the size of the returned $(\T,c)$-Sparsifier, and the runtime of \phisparsify respectively. 

\subparagraph{Correctness.}
According to \Cref{lem:contraction-invariant} and \Cref{cor:batch-contraction-invariant}, the modified $\Gaux$ is always 
the pruned auxiliary graph of $G/E'$ at the beginning of each iteration of \Cref{line:enum-edge}.
By \Cref{lem:essential-auxiliary}, only non-essential hyperedges will be contracted.
Therefore, by \Cref{def:essential-edges}, the values of $p$-mincuts for all useful partition $p$ 
are not affected and thus the returned hypergraph is a $(\T,c)$-sparsifier.  

\subparagraph{Size of the $(\T,c)$-Sparsifier.} 
Since all essential hyperedges cannot be contracted as they will affect $(A,\T\setminus A)$-mincut value for some $A \subset \T$, any $(\T, c)$-sparsifier of $G$ must include all essential hyperedges in $G$. By the existence theorem (\Cref{thm:existence-of-sparsifier}), the number of essential hyperedges in $G$ is at most $O(|\T|c^3)$. By \Cref{lem:essential-edges-remain-essential}, after contracting all non-essential hyperedges detected by \phisparsify, all the remaining hyperedges are essential. Therefore, the number of hyperedges in the $(\T,c)$-sparsifier returned by \phisparsify is at most $O(|\T|c^3)$.

\subparagraph{Runtime of \phisparsify.}
\Cref{line:phi-sparsify-contruct-auxiliary-graph} runs in $n(rc\phi^{-1})^{O(rc)}$ time by \Cref{lem:construct-pruned-auxiliary-graph}.
Since $|P'_e|\le |N(e)|$ for any $e\in E_0$, we know that \Cref{line:phi-sparsify-find-all-partition}
contributes a total runtime of at most
\begin{align*}
\sum_{e\in E_0} |N(e)| 
\le c|C_0| = n(rc\phi^{-1})^{O(rc)}.
\end{align*}
To implement \Cref{line:phi-sparsify-check-partition}, we create a hash table marking all vertices in $N(e)$ and for each $x\in P'_e$ the algorithm  iterates through $N(x)$ until the first neighbor that is not in $N(e)$. The whole process takes $O(|N(e)|)$ time. Hence, the total runtime contributed by \Cref{line:phi-sparsify-check-partition} takes $n(rc\phi^{-1})^{O(rc)}$ as well.
Now, since each cut in $C_0$ will be removed at most once, so the total runtime contributed by \Cref{line:phi-sparsify-contruct-auxiliary-graph} is at most $c|C_0|=n(rc\phi^{-1})^{O(rc)}$.

Finally, contracting all hyperedges in $E\setminus E'$ on $G$ takes linear time in the total hypergraph size $O(p)$.
Hence, the total time for $\phisparsify$ is $O(p + n(rc\phi^{-1})^{O(rc)})$.
\end{proof}

\subsection{The Almost-linear-time Algorithm of $(\T, c)$-Sparsifiers}
\label{sec:Sparsification-algorithm}
Finally, we combine all the tools from previous subsections and formally prove \Cref{thm:main} (1).

To construct a $(\T, c)$-sparsifier for a hypergraph $G = (V,E)$ with $O(|\T|c^3)$ hyperedges under the divide-and-conquer scheme,
we first split the hypergrpah $G$ into $\phi$-expanders $\{G_i = (V_i,E_i)\}$ via \expdec.
All boundary hyperedges in $\cup_i \partial V_i$ will be separated, and we add two anchor vertices per separated hyperedge.

For each $\phi$-expander, we enumerate all the connected terminal cuts with value at most $c$ via \Cref{alg:enumerate-cuts}. Then, we construct the auxiliary graph and prune it. Next we contract non-essential hyperedges with the help of the pruned auxiliary graph. Lastly, we glue all the sparsifiers of the $\phi$-expanders together by replacing all anchored separated hyperedges with the original hyperedges. 

The above procedure described in the previous two paragraphs is one iteration of sparsifying. The following algorithm \Cref{alg:sparsify-fast} repeats the procedure until there has been $\log m$ iterations. Finally, \Cref{alg:sparsify-fast} returns a $(\T,c)$-sparsifier with $O(|\T|c^3)$ hyperedges.

\begin{algorithm}[htbp]
\DontPrintSemicolon
\KwIn{hypergraph $G = (E,V)$ with rank $r$, terminal set $\T$, threshold parameter $c$, constant $C'$.}
\KwOut{a $(\T,c)$-sparsifier $H$.}
$H \gets G$.\\
$iter \gets 0$. \tcc{Number of iterations of the following while-loop.}
\Do{$iter < \log{m}$.}{
    $G \gets H$ \\
    $\phi^{-1} \gets 4C'rc^4 \log^3{n}$. \\
    $\{V_i \}_{i=1}^t \gets \expdec(G,\phi)$. \label{line:expander-decomposition}\\
    $G' \gets G$ \tcc{Anchored sub-hypergraphs will be separated from $G'$ one by one.} 
    \For{each $i=1,2,\ldots,t$}{
        Apply the divide step in \Cref{alg:divide-and-conquer} to $G'$ with terminal $\T$ and bipartition $(V_i, \bigcup_{\ell = i+1}^{t} V_\ell) $, and get $G_i \gets \hat{G'}[V_i]$ and
        $G' \gets \hat{G'}[\bigcup_{\ell=i+1}^{t} V_\ell]$.\\
        (For each boundary hyperedge $e$ with anchor vertices $\anchor_{e, 3}$ and $\anchor_{e, 4}$ created on the $\bigcup_{\ell=i+1}^t V_\ell$ side, we assign
        both $\anchor_{e, 3}$ and $\anchor_{e, 4}$ to an arbitrary $V_j$ such that $j>i$ and $e\cap V_j\neq \emptyset$.)
    }
    $\{H_i\}_{i=1}^t \gets \{ \phisparsify(G_i, V_i \cap \T, \phi, r,c)\}_{i=1}^t$. \label{line:sparsify} \tcc{The conquer step.}
    $H \gets H_t$ \tcc{Each sparsifier $H_i$ will be merged with $H$ one by one.}
    \For{each $i=t-1,\ldots, 1$}{
        Apply the combine step in \Cref{alg:divide-and-conquer} to merge $H_i$ with $H$. That is, all anchor vertices introduced at the divide step are removed and all separated hyperedges are replaced by the boundary hyperedges before separating $V_i$ from $\bigcup_{\ell=i}^t V_\ell$.
        \label{line:combine-sparsifiers}
    }
    $iter = iter + 1$.
}
\Return $H$.
\caption{\algfast$(G, r, \T, c, C')$}
\label{alg:sparsify-fast}
\end{algorithm}

With \Cref{alg:sparsify-fast}, we are ready to prove \Cref{conclusion}.

\begin{proof}
We show the correctness of \Cref{alg:sparsify-fast}, the size of the sparsifier the algorithm returns, and the time complexity.

\subparagraph{Correctness.} We note that any $\phi$-expander graph $G$, adding new (anchor) vertices to any existing hyperedge on $G$, is still an $\phi$-expander.
Each iteration of \Cref{alg:sparsify-fast} adopts the divide-and-conquer scheme described by \Cref{alg:divide-and-conquer} is justified by \Cref{lem:combine-sparsifiers}. Similarly, sparsifying $\phi$-expanders using \phisparsify described by \Cref{alg:phi-sparsify} is justified by \Cref{lem:sparsifier-validity-and-time}.

\subparagraph{Size of the sparsifier.}
We first focus on how much \Cref{alg:sparsify-fast} can sparsify a hypergraph in one iteration. Since \expdec guarantees there are at most $O(m\phi \log^3{n})$ inter-cluster hyperedges, in one iteration, there are at most $O(rm\phi \log^3{n})$ separated hyperedges. Therefore, the total number of all terminals (including the originally given terminals and additional anchor terminals) is bounded by $|\T| + C'rm \phi \log^3{n}$ for some constant $C'$. In each iteration, let the algorithm set $\phi = 1/(4C'rc^4 \log^3{n})$, so $\phi (C'rc^4 \log^3{n} + C' \log^3{n}) \leq \frac{1}{2}$. Denote the total number of hyperedges after iteration $i$ as $m_i$, we have
\begin{align*}
  m_{i+1} & = (|\T| + C'm_ir\phi \log^3{n})c^3 + C'm_i\phi \log^3{n} \\
  & = |\T|c^3 + m_i\phi (C'rc^3 \log^3{n} + C' \log^3{n}) \\
  & \leq |\T|c^3 + \frac{m_i}{2}.
\end{align*}
Notice that $m_0 = m$. Therefore, after $\ell$ iterations, the number of hyperedges is less than $$|\T|c^3 \sum_{i = 0}^{\ell-1}(\frac{1}{2^i}) + \frac{m}{2^{\ell}}.$$

After $O(\log{m})$ iterations, the algorithm returns a $(\T,c)$-sparsifier with $O(|\T|c^3)$ hyperedges.

\subparagraph{Time complexity.} In each iteration, the time is dominated by \expdec and \phisparsify, which take $O(p^{1+o(1)})$ and  $O(p + n(rc\phi^{-1})^{O(rc)})$ time repectively by \Cref{lem:expander-decompose-for-hypergraphs} and \Cref{lem:sparsifier-validity-and-time}. Notice that there are $\log m$ iterations, but in each iteration the number of hyperedges are halved.
Therefore we obtain
the total running time $O(p^{1+o(1)} + n(r^2c^5\log^3{n})^{O(rc)} \log{m})$ $=$ $O(p^{1+o(1)} + n(rc \log{n})^{O(rc)}\log m)$.
\end{proof}

This conclude the proof of \Cref{thm:main} (1).
For Part (2) of \Cref{thm:main}, we observe that, by working with appropriate notions of expansion in hypergraphs, the construction of Liu \cite{Liu20} naturally extends to hypergraphs. We give the proof for completeness in  \Cref{sec:part-2-of-main}.

\section*{Acknowledgement}

We thank Sorrachai Yingchareonthawornchai, 
Yang Liu, Yunbum Kook, and Richard Peng for the discussion that inspires the notion of the pruned auxiliary graph in this paper. 
We also thank anonymous reviewers for their valuable comments.

\bibliographystyle{alpha}
\bibliography{main}

\appendix

\section{Omitted Content from \Cref{sec:structural-properties}}

\subsection{Proof of \Cref{lem:combine-sparsifiers} Part 2}\label{appendix:part-2-proof}

In this subsection we show that $\mincut_H^c(\pi(A), \pi(B))\ge \mincut_G^c(A, B)$.

We first assume that $(Y, \pi(V)\setminus Y)$ is a $(\pi(A), \pi(B))$-mincut of value at most $c$.

Let $S=E_G(V_1, V_2)$.
We know that $H$ is defined by $H_1\cup H_2 \cup S - \Sep(S, V_1, V_2)$, where $H_1$ is a $(\T_1, c)$-sparsifier of $\hat{G}[V_1]$ and $H_2$ is a $(\T_2, c)$-sparsifier of $\hat{G}[V_2]$.
Again, we define the cuts $Y_1$ and $Y_2$ respectively.
We first add $Y\cap \pi(V_1)$ to $Y_1$, then for each separated hyperedge $\pi(e)$ where $e\in S$ there are three cases for adding the terminal vertices into $Y_1$: we add $\{\pi(\anchor_{e, 1}), \pi(\anchor_{e, 2})\}$ to $Y_1$ if $\pi(e)\subseteq Y$; we add $\pi(\anchor_{e, 1})$ to $Y_1$ if $\pi(e)\not\subseteq Y$ and $\pi(e)\cap Y\neq \emptyset$; we add nothing if $\pi(e)\cap Y=\emptyset$. Similarly we define $Y_2$ by adding $Y\cap \pi(V_2)$ and corresponding anchor vertices.

Now we have two cuts $(Y_1, \pi(\hat{V}_1)\setminus Y_1)$ and $(Y_2, \pi(\hat{V}_2)\setminus Y_2)$ with values $y_1$ and $y_2$. Each of them has value at most the value to the cut $(Y, \pi(V)\setminus Y)$ so we know that both $y_1, y_2\le c$.
They are $(\pi(\Aextend|_{V_1\cup A_1}), \pi(\Bextend|_{V_1\cup A_1}))$-cut and $(\pi(\Aextend|_{V_2\cup A_2}), \pi(\Bextend|_{V_2\cup A_2}))$-cut respectively.

By definition of $(\T_1, c)$-sparsifier and $(\T_2, c)$-sparsifier,
there exists a $(\Aextend|_{V_1\cup A_1}, \Bextend|_{V_1\cup A_1})$-mincut $(X_1, \hat{V}_1\setminus X_1)$ and a $(\Aextend|_{V_2\cup A_2}, \Bextend|_{V_2\cup A_2})$-mincut $(X_2, \hat{V}_2\setminus X_2)$ whose cut values are $x_1\le y_1$ and $x_2\le y_2$.

Now, notice that $y_1+y_2 = |\partial Y| + |\partial Y\cap \pi(S)|$. Here we emphasize that $\partial Y\cap \pi(S)$ may be a multiset, as there could be many hyperedges being contracted into the same subset of vertices.
Again, let $X=X_1\cup X_2\cup S -\Sep(S, V_1, V_2)$, we know that $|\partial X|  = x_1+x_2-|(\partial X)\cap S|$. 
Since all edge $e\in S$ such that $\pi(e)\in \partial Y$ were not contracted (they must be essential as such edges appear in every $(\Aextend|_{V_i\cup A_i}, \Bextend|_{V_i\cup A_i})$-mincut for both $i\in \{1, 2\}$), we know that $|(\partial X)\cap S|\ge |\partial Y\cap \pi(S)|$. Hence, we have
\begin{align*}
    \mincut_G^c(A, B) &\le |\partial X|\\
    &= x_1+x_2-|(\partial X)\cap S| \\
    &\le y_1+y_2-|(\partial Y)\cap \pi(S)|\\
    &=|\partial Y| \\
    &= \mincut_H^c(\pi(A), \pi(B))
\end{align*}
as desired. \hfill $\square$

\subsection{Proof of \Cref{lem:gammoid}}\label{appendix:proof-of-gammoid}

\newcommand{\Gsplit}{G_{\mathsf{split}}}
\newcommand{\M}{\mathcal{M}}
\newcommand{\rank}{\ensuremath{\mathrm{rank}}}

\Lemmagammoid*

\begin{proof}[Proof to \Cref{lem:gammoid}]
Let $k=|\T|$.
It suffices to show that whenever $G$ has $6k c^2+1$ edges,
there exists one hyperedge $e$ so that $G/e$ is a $(\T, c)$-sparsifier of $G$ and that $\T$ is still $(5c, c)$-edge-unbreakable in $G/e$.

Intuitively,
to find such an hyperedge,
we start with
reducing the $c$-thresholded hyperedge connectivity into vertex $c$-connectivity
by carefully constructing the graph $\Gsplit$ from $G$.
Then, we apply the powerful technique discovered by Kratsch and Wahlstr{\"{o}}m~\cite{KratschW20} on $\Gsplit$,
which defines a representable matroid $\M_1\oplus\M_2\oplus\M_3$ based on $\Gsplit$ and its two variants $\Gsplit'$ and $\Gsplit''$. 
Followed by a carefully chosen collection of subsets $\mathcal{J}$, 
there exists a representative set $\mathcal{J}^*\subseteq \mathcal{J}$ of size at most $\rank(\M_1)\rank(\M_2)\rank(\M_3) \le 6k c^2$.
Moreover, the elements of $\mathcal{J}$ corresponds to vertices and hyperedges in $G$.
Since $G$ has at least $6k c^2+1$ hyperedges, at least one hyperedge $e$ on $G$ is not \emph{essential}, and hence contractable.
The above process for finding a contractable edge $e$  will be repeated until there are at most $6k c^2$ hyperedges left.

Let $G=(V, E)$ be a hypergraph and let $\T\subseteq V$ be a set terminal vertices.
We first construct the undirected \emph{incidence} bipartite graph $\Ginc=(\Vinc, \Einc)$ as the following:
\begin{align*}
    \Vinc &= V\cup E\\
    \Einc &= \{ \{v, e\}\ |\ v\in V, e\in E, \text{ and } v\in e\}.
\end{align*}
Now, we perform a reduction from hyperedge cuts on $G$ to $c$-thresholded vertex cuts on $\Ginc$ similar to \cite{Liu20} and construct a directed graph $\Gsplit$ based on the incidence graph $\Ginc$: for each vertex $v\in \Vinc\cap V$ we split $v$ into a $(c+1)$-clique on $\Ginc$, then we replace each undirected edge $\{v, e\}$ in $\Einc$ with two directed edges $(v, e)$ and $(e, v)$.
In \Cref{claim:gammoid:reduction-to-vertex-cuts} we show that all $c$-thresholded hyperedge mincuts on $G$ correspond to some $c$-thresholded vertex mincuts on $\Gsplit$.
The definition of $\Gsplit$ inherently defines a mapping $\pi: V\cup E \to \Vsplit$ where for each $v\in V$, $\pi(v)$ is assigned to be any of the $c+1$ vertices split from $v$.

\begin{claim}\label{claim:gammoid:reduction-to-vertex-cuts}
Let $A, B\subseteq V$ with $\mincut_G(A, B)\le c$.
Let $\Asplit=\pi(A)$ and $\Bsplit=\pi(B)$.
Then the minimum vertex cut $\Csplit\subseteq \Vsplit$ that separates $\Asplit$ and $\Bsplit$ on $\Gsplit$ has size $|\Csplit|=\mincut_G(A, B)$.
\end{claim}

\begin{proof}[Proof of \Cref{claim:gammoid:reduction-to-vertex-cuts}]
First, it is straightforward to check that $|\Csplit| \le \mincut_G(A, B)$: let $(X, V\setminus X)$ be an $(A, B)$-mincut on $G$ with $|\partial X|=\mincut_G(A, B)$. Hence, $\pi(\partial X)$ is a vertex cut that separates $\pi(A)$ and $\pi(B)$ on $\Ginc$ of the same size.

Now we show that $|\Csplit| \ge \mincut_G(A, B)$. Notice that since $|\Csplit|\le \mincut_G(A, B)\le c$, $\Csplit$ does not contain split vertices from $V$. Hence, $\pi^{-1}(\Csplit)$ is a $(A, B)$-cut on $G$, and the result follows.
\end{proof}

Now, we define $\Gsplit'$ and $\Gsplit''$ from $\Gsplit$.
Both $\Gsplit'$ and $\Gsplit''$ are initialized as $\Gsplit$.
For each vertex $v\in \Gsplit$, we create an sink-copy $v'$ in $\Gsplit'$ and 
for each edge $(u, v)\in \Gsplit$ we add an edge $(u, v')$ to $\Gsplit'$.
Similarly, for each vertex $v\in \Gsplit$, we create an source-copy $v''$ in $\Gsplit''$ and for each edge $(v, w)\in \Gsplit$ we add an edge $(v'', w)$ to $\Gsplit''$.
In order to define gammoid we specify the terminal set $\Tsplit = \pi(\T)$ on both $\Gsplit'$ and $\Gsplit''$.

\subparagraph{Matroid Construction.}
We define a representable matroid $\M = \M_1\oplus \M_2\oplus \M_3$ as follows.
\begin{itemize}[itemsep=0pt]
\item Let $\M_1=(\Vsplit, {\Vsplit \choose {\le c}})$ be the uniform matroid of rank $c$ over the vertices on $\Vsplit$.
\item Let $\M_2$ be the gammoid over $\Gsplit'$ with respect to the terminal set $\Tsplit$.
\item Let $\M_3$ be the gammoid over $\Gsplit''$ with respect to the terminal set $\Tsplit$ but restricting the rank to $6c$ (getting rid of all independent sets larger than $6c$).
\end{itemize}

$\M_1$ has rank $c$, $\M_2$ has rank at most $k$, and $\M_3$ has rank at most $6c$.
Let $\J = \{ (v, v', v'') \ |\ v\in \Vsplit \cap E\}$ be a collection of subsets in $\M$.
According to the theorem of representative set (\Cref{lem:matroid-representative-set}),
there exists a representative set $\J^*$ of size at most $\rank(\M_1)\rank(\M_2)\rank(\M_3) \le 6k c^2$.

\begin{lemma}[\cite{Liu20,Marx09,Lovasz77}]\label{lem:matroid-representative-set}
Let $\M_i=(S_i, \mathcal{I}_i)$ be representable matroids over a field $\mathbb{F}$ for $1\le i\le p$,
and let $\M=\M_1\oplus \M_2\oplus \cdots \oplus \M_p$ be their direct sum. Define
\[
S_1\times S_2\times \cdots \times S_p = \{\{x_1, x_2, \ldots, x_p\}\ :\ x_i\in S_i\}.
\]
That is, the collection of subsets of $S$ that have exactly one element from each $S_i$.
Then every subset $\mathcal{J}\subseteq (S_1\times S_2\times \cdots \times S_p)$
has a representative set $\mathcal{J}^*$ of size at most 
$\prod_{i=1}^p \rank(\M_i)$, and $\mathcal{J}^*$ is computable in time $|\mathcal{J}|^{O(1)}$.
\end{lemma}

\subparagraph{Essential Vertices.} We say that a vertex $v$ is \emph{essential} if there exists a bipartition $(\Asplit, \Bsplit)$ of $\Tsplit$ with the minimum vertex cut between $\Asplit$ and $\Bsplit$ at most $c$ and that $v$ appears in every minimum vertex cut between $\Asplit$ and $\Bsplit$.

Now the key observation is that for each essential vertex $v$, the corresponding triple $(v, v', v'')$ 
must appear in any representative set $\J^*$.
Let $v$ be an essential vertex with respect to the bipartition $(\Asplit, \Bsplit)$.
By the definition of $(5c, c)$-edge-unbreakable we may assume that $|\Bsplit| < 5c$.
Fix any minimum $(\Asplit, \Bsplit)$-cut $\Csplit$, and consider the following independent set $I=(\Csplit\setminus \{v\}, \Asplit\cup \Csplit, \Bsplit\cup \Csplit)$ in $\M$.
Now, if there is another $(u, u', u'')\in \J^*$ with $u\notin \Csplit$ then we must have $(\Csplit\setminus \{v\}\cup \{u\}, \Asplit \cup \Csplit\cup \{u'\}, \Bsplit\cup \Csplit\cup \{u''\}) \notin \mathcal{M}$ because of the following reason: (1) the first component disallows $u\in \Csplit$, (2) since $u\notin \Csplit$, $u$ is not essential, and it is impossible to have vertex disjoint paths going from $\Tsplit$ to $\Asplit\cup \Csplit\cup \{u'\}$ and vertex disjoint paths going from $\Tsplit$ to $\Bsplit \cup \Csplit\cup \{u''\}$ at the same time (otherwise $\Csplit$ does not separate $\Asplit$ and $\Bsplit$).
In addition, since $v$ is essential, there no other mincut closer to $\Asplit$ or closer to $\Bsplit$ that does not include $v$. Hence, adding $(v, v', v'')$ to $I$ will result in independent set.
Therefore, $(v, v', v'')$ is required to appear in any representatitve set. By \Cref{lem:matroid-representative-set}, the number of essential hyperedges is at most $6kc^2=O(kc^2)$ as desired.
\end{proof}

We can extend the proof to \Cref{lem:gammoid} when we allow $(d, c)$-edge-unbreakable terminal sets with an arbitrary value $d$. The same proof holds except that we restrict $\M_3$ to have rank at most $(d+c)$. So the total number of essential edges are at most $O(|\T|c(d+c))$ and they can be constructed in polynomial time.

\begin{corollary}\label{cor:dcedgeunbreakable}
Let $G=(V, E)$ be a hypergraph and let $\T\subseteq V$ be a set of degree-1 terminals.
If $\T$ is $(d, c)$-edge-unbreakable on $G$, then there exists a subset $E'\subseteq E$ with $O(|\T|c(d+c))$ hyperedges, such that $G/(E-E')$ is a $(\T, c)$-sparsifier of $G$.
Moreover, the set $E'$ can be computed in $\text{poly}(m, n)$ time.
\end{corollary}

\section{Existence and Uniqueness of $A$-minimal $(A,B)$-mincut} \label{sec:uniqueness-of-A-minimal-cut}

In this section, we show the proof of \Cref{lem:uniqueness-of-A-minimal-mincut}

\Lemmauniqueminimalcut*

\begin{proof}
Let $f_G: 2^V \to \mathbb{R}$ be a function and $f_G(X) = |\partial X|$. We write $f_G$ as $f$ for simplicity and claim that $f$ is a submodular function: For arbitrary $X, Y \subseteq V$, we separate the hypergraph into disjoint four parts $Y\setminus X$, $X\setminus Y$, $X \cap Y$, and $V\setminus(X\cup Y)$. By identifying $\partial X$ and $\partial Y$ as disjoint unions of crossing hyperedges between the four parts, we have  
\begin{align*}
f(X) + f(Y) & = |\partial X| + |\partial Y| \\
& = (|E(X\setminus Y, V \setminus (Y \cup X))| + |E(X\setminus Y, Y\setminus X)| \\
& \quad + |E(X \cap Y, V \setminus (Y \cup X))| + |E(X \cap Y, Y\setminus X)|) \\
& \quad + (|E(Y\setminus X,V\setminus(X \cup Y))| + |E(Y\setminus X,X\setminus Y)| \\
& \quad + |E(Y\cap X, V\setminus(X \cup Y))| + |E(Y \cap X, X \setminus Y)|) 
\end{align*}

On the other hand, similarly we have
\begin{align*}
f(X\cup Y) + f(X \cap Y) & = |\partial{(X \cup Y)}| + |\partial{(X \cap Y)}|  \\
& = (|E(X\setminus Y,V\setminus (X\cup Y))| + |E(X\cap Y, V\setminus (X\cup Y))| \\
& \quad + |E(Y\setminus X, V\setminus (X\cup Y))|) + (|E(X\cap Y, Y\setminus X)| \\ 
& \quad + |E(X\cap Y, X\setminus Y)| + |E(Y\cap X, V\setminus (X\cup Y))|) 
\end{align*}

Therefore, $f(X) + f(Y) - f(X\cup Y) - f(X\cup Y) = 2|E(X\setminus Y, Y\setminus X)| \geq 0$ and finally get 
$$f(X) + f(Y) \geq f(X\cup Y) + f(X\cap Y)$$
$f$ is submodular function.

\newcommand{\ampar}{$A$-minimal partition}

If $(X,V\setminus X)$ and $(Y,V\setminus Y)$ are $(A,B)$-mincuts, by the submodular property, $f(X\cup Y), f(X\cap Y)$ are also equal to the $(A,B)$-mincut size, which indicates $(X\cup Y, V\setminus (X\cup Y))$ and $(X\cap Y, V\setminus (X\cap Y))$ are also $(A,B)$-mincuts. Giving that $(X\cap Y, V\setminus (X\cap Y))$ is a $(A,B)$-mincut, the $X^*$ of the \ampar $(X^*,V\setminus X^*)$ is simply the intersection of all $X \subset V$ such that $(X,V \setminus X)$ is a $(A,B)$-mincut. The existence and uniqueness of \ampar \ are verified by construction.
\end{proof}

\section{Use {\normalfont{\locmincut}} to find \amcut} \label{sec:local-min-cut}

In this section we describe how to apply~\cite{forster2020computing} to prove~\Cref{thm:local-min-cut}.

\Cref{lem:use-amcut-to-identify-useful-partition} reduces the problem of testing whether a terminal partition $(A,\T\setminus A)$ is useful to checking whether the \tmcut is connected. In this subsection, we show how to further reduce the problem from a hypergraph setting to a normal directed graph setting, so that results in \cite{forster2020computing} can directly solve it.

We first convert the hypergraph into an \emph{incidence graph},

\begin{definition}[Incidence Graph of a Hypergraph]\label{def:incidence-graph}
We say a normal graph $G = (V',E')$ is the incidence graph of a hypergraph $H = (V,E)$ if $V' = V \cup E$ and $E' = \{(v,e)| v \in V, e\in E, v \in e \ in \ H\}$.
\end{definition}

We note that $G$ is a bipartite graph.
The incidence graph is a normal graph and also inherits the connectivity of the hypergraph. 

Then we replace each vertex by an in-vertex and an out-vertex in $G$. We connect a pair of in-vertex and out-vertex with 1 (c + 1) edge if the vertex represents a hyperedge (vertex) in $H$. Lastly we replace each edge by two directed edge of reverse directions. The above reduction is a common way to reduce edge to vertex connectivity and undirected to directed setting, except that we add $(c+1)$ multi-edges between certain in-vertices and out-vertices to prevent all $v \in V$ from being cut vertex in $G$. 

From the above manipulation to $G$, we get a normal directed graph $G'$. An important observation is that the \amcut in $G$ exactly corresponds to the \amcut in $H$. Therefore, the problem of finding the \amcut in $H$ is reduced to finding the \amcut in $G'$, which can be directly solved by the following result in \cite{forster2020computing}.

\begin{lemma}[\cite{forster2020computing}]
\label{lem:local-min-cut}
Let $G'$ be a normal graph.
Suppose that the $A$-minimal $(A,B)$-mincut $S$ has size at most $c$ and volume $\vol(S)=\sum_{v \in S}\deg_{G'}(v)\le \nu$ in a normal directed graph. 
Then, there is an algorithm that finds $S$ in $O(\nu c^2 \log n)$ time with high probability. 
\end{lemma}

\begin{proof}[Proof of \Cref{thm:local-min-cut}.]
Notice that each vertex in $G'$ may have $c$ more incident edges than $G$. Hence, the volume of the same \tmcut in $G'$ is boosted up by at most a factor of $c$.
\end{proof}

\section{Proof of Part (2) of \Cref{thm:main}}
\label{sec:part-2-of-main}

In this section we extend Liu's result~\cite{Liu20} and prove the following theorem.

\begin{theorem}\label{thm:liu-polytime-sparsifier}
Let $G$ be a hypergraph with $n$ vertices, $m$ hyperedges, and a set $\T\subseteq V$ of degree-1 terminals with $|\T|\le kc$.
There exists an algorithm that constructs a $(\T, c)$-sparsifier of $G$ with $O(kc^3\log^{1.5}(kc))$ hyperedges in $\mathrm{poly}(m, n)$ time.
\end{theorem}

\newcommand{\polytimesparsifier}{\textsc{PolyTimeSparsifier}}
\newcommand{\scn}{\beta}

The main runtime bottleneck to \algslow (\Cref{alg:sparsify-slow}) is in \Cref{line:split-condition}.
That is, 
it is not known if there is an polynomial time algorithm that tests whether the current terminal set $\T'$ is $(5c, c)$-edge-unbreakable in the subproblem $G'$.
Fortunately, this problem can be modeled as a sparsest cut problem with a non-uniform demand.
If we allow a $\text{polylog}(k)$ factor on the number of hyperedges in the final $(\T, c)$-sparsifer,
we don't have to apply divide and conquer framework until we have a $(5c, c)$-edge-unbreakable terminal set.
For example, a $(c\log^{1.5}(kc), c)$-edge-unbreakable terminal set suffices to obtain a sparsifier of $O(kc^2\log^{1.5}(kc))$ hyperedges by \Cref{cor:dcedgeunbreakable}.

Let us now describe the sparsest cut problem (with non-uniform demand) we are reducing to.
We first define \emph{terminal expansion}:

\begin{definition}[terminal expansion]
A hypergraph $G=(V, E)$ with terminal set $\T\subseteq V$ is a $\phi$-terminal expander if
\[
\frac{|\partial X|}{\min\{ |X\cap \T|, |(V\setminus X)\cap \T| \}}\ge \phi \ \ \text{for any subset of vertices $X\subseteq V$.}
\]
\end{definition}

We remark that whenever $\phi=O(\log^{-1.5}(kc))$, if the graph $G$ with a terminal set $\T$ is a $\phi$-terminal expander, then $\T$ is $(O(c\log^{1.5}(kc)), c)$-edge-unbreakable.
Therefore, the base cases are solved.
All we have to do now is to find a proper way to check if the hypergraph in a subproblem is a $\phi$-terminal expander. Furthermore, if it is not the case, we should be able to identify a good-enough sparse cut $(V_1, V_2)$ that enables the divide and conquer approach.

Let us now translate the edge cuts in this hypergraph $G$ to vertex cuts in the incidence graph representation (see \Cref{def:incidence-graph}).
On this incidence graph $\Ginc=(\Vinc=V\cup E, \Einc)$ we define two functions over $\Vinc$:
\[
\begin{array}{ll}
\pi_1(x)=\begin{cases}
1 & \text{if $x\in E$,}\\
\infty & \text{otherwise, $x\in V$.}
\end{cases}&
\pi_2(x)=
\begin{cases}
1 & \text{if $x\in \T$,}\\
0 & \text{otherwise.}
\end{cases}
\end{array}
\]
It is straightforward to check that among all vertex partition $(A, B, S\subseteq E)$,
\[
\mathsf{OPT} := \min_{A',B',S'}
\frac{\pi_1(S')}{\pi_2(A'\cup S')\pi_2(B'\cup S')}
\] is 2-approximating the terminal expansion.
The following theorem gives a sparsest vertex cut solver that returns an $O(\sqrt{\log (kc)})$-approximated\footnote{The approximate ratio is actually $\sqrt{\log |\text{supp}(\pi_2)|}$. Notice that the support of $\pi_2$ is $|\T|\le kc$.} sparse cut in polynomial time.
\begin{theorem}[\cite{FeigeHL08} Theorem 3.12]\label{thm:fhl08}
Given a graph $G=(V, E)$ and vertex weights $\pi_1, \pi_2: V\to \mathbb{R}_+$, there exists a polynomial-time algorithm which computes a vertex separator $(A, B, S)$ for which \[
\frac{\pi_1(S)}{\pi_2(A\cup S)\pi_2(B\cup S)} \le \scn \cdot \mathsf{OPT},
\]where $\scn=O(\sqrt{\log (kc)})$ is the approximation ratio.
\end{theorem}

Now, we can use an approximate sparsest cut solver that either (1) certifies that the $\T'$ is $(c \log^{1.5}(kc), c)$-edge-unbreakable, or (2) finds a violating partition $(V_1, V_2)$ such that
the number of crossing hyperedges
is at most $\frac{1}{2 \log (kc)}|\T'\cap V_i|$ for both $i\in \{1, 2\}$.
Notice that it implies that in the subproblems the amount of terminal vertices are increased by only a $1/\log (kc)$ factor since we add 2 anchor vertices to each separated hyperedge.
We summarize the algorithm in \Cref{alg:sparsify-polytime}.

\begin{algorithm}[h]
\DontPrintSemicolon
\nonl\polytimesparsifier$(G, \T, c)$\\
\KwIn{hypergraph $G = (V, E)$, terminal set $\T$, constant $c$. }
\KwOut{a $(\T,c)$-sparsifier $H$.}
$\phi \gets 1/(4 \scn \log (kc))$.\\
Run a sparse cut solver and obtain a bipartition $(V_1, V\setminus V_2)$.\\
Compute terminal expansion $\hat{\phi} := \frac{|\partial V_1|}{\min\{|V_1\cap \T|, |V_2\cap \T|\}}$.\\ 
\eIf{$\hat{\phi} \ge 2\scn \phi$}{
\tcp{guaranteed $\phi$-terminal expander.}
\Return a $(\T, c)$-sparsifier using \Cref{cor:dcedgeunbreakable}.
}{
\tcp{found a sparse cut.}
Apply divide and conquer framework (\Cref{alg:divide-and-conquer}) on $(V_1, V_2)$.
}
\caption{\textsc{PolyTimeSparsifier}}
\label{alg:sparsify-polytime}
\end{algorithm}

\begin{proof}[Proof of \Cref{thm:liu-polytime-sparsifier}]
The correctness is entirely based on the divide and conquer framework (\Cref{lem:combine-sparsifiers}). 
The runtime is $\text{poly}(mn)$ because of \Cref{cor:dcedgeunbreakable} and \Cref{thm:fhl08}.
Now, in order to give an upper bound of hyperedges from the returned $(\T, c)$-sparsifier, it suffices to bound the total number of terminal vertices over all base cases.

Let $f(x)$ be an upper bound of total number of terminal vertices over all base cases when started with a terminal set of $x$ vertices.
We first assume that $f$ is a non-decreasing function and $f$ is concave. So, according to the divide and conquer criteria, we know that if a terminal set of $x$ vertices is split into $y$ terminal vertices and $x-y$ terminal vertices (without loss of generality $y\le x-y$), 
the number of hyperedges crossing the cut is at most \[
\hat{\phi}y \le 2\scn\phi y = \frac{1}{2\log (kc)} y.
\]
Hence, by adding 2 anchor vertices to each separated hyperedge we have 
\begin{align*}
f(x) &\le f\left(y+2\textstyle\frac{1}{2\log (kc)}y\right) + f\left(x-y + 2\textstyle\frac{1}{2\log (kc)}y\right)\\
&\le 2f\left(\left(1+\textstyle\frac{1}{\log (kc)}\right) \frac{x}{2}\right)\tag{by concavity and non-decreasing}\\
&\le 2^{\log x} f\left(\textstyle\left(1+\frac{1}{\log (kc)}\right)^{\log x}\right) \tag{induction on $x$}\\
&\le 2^{\log x} f(\exp(1)) \tag{as long as $x \le kc$}\\
&= \Theta(x). 
\end{align*}

So, choosing $f:[1,kc]\to \mathbb{R}$ to be a linear function serves as an upper bound. By \Cref{cor:dcedgeunbreakable} the total number of hyperedges in the returned $(\T, c)$-sparsifier is then $O(|\T|c^2\log^{1.5}(kc)) = O(kc^3\log^{1.5}(kc))$.
\end{proof}

\end{document}